\documentclass[11pt,a4paper]{article}

\usepackage{tikz}
\usepackage{authblk}
\usepackage[left=1in, right=1in, top=1in, bottom=1in]{geometry}
\usepackage{amsthm}
\usepackage{amsfonts}
\usepackage{mathrsfs}
\usepackage[pdftex]{hyperref}
\usepackage{amssymb}
\usepackage{amsmath}
\usepackage{cleveref}
\usepackage{xcolor}
\usepackage{graphicx}
\usepackage{xparse}
\usepackage{subcaption}
\PassOptionsToPackage{dvipdfmx}{graphicx}
\usepackage[utf8]{inputenc}
\usepackage[T1]{fontenc}
\usepackage[makeroom]{cancel}

\usepackage[toc,page]{appendix}

\usepackage{macros}

\begin{document}

\title{Localization Dichotomy for Hall Insulators on the Lattice}

\author{Alex Bols \thanks{email: \href{abols01@phys.ethz.ch}{abols01@phys.ethz.ch}}}
\affil{Institute for Theoretical Physics, ETH Z{\"u}rich}

\maketitle

\begin{abstract}
    We prove the localization dichotomy conjecture for non-periodic Hall insulators on the lattice. That is, we show that the Fermi projector $P$ of a gapped exponentially local Hamiltonian on $l^2(\Z^2) \otimes \C^n$ admits an exponentially localized Wannier basis if and only if $P$ has vanishing Chern number.
\end{abstract}

\setcounter{tocdepth}{2}
\tableofcontents

\section{Introduction} \label{sec:introduction}

It is well-known that the Fermi projection of a gapped and exponentially local single-particle Hamiltonian is again exponentially local, in the sense that its off-diagonal matrix elements decay exponentially fast in the distance. A stronger property for such a projector $P$ is the existence of an \emph{exponentially localized Wannier basis}, an orthonormal basis of $\Ran P$ comprised of states each of which is exponentially localized near some localization center. The existence of well-localized Wannier bases underpins a slew of efficient numerical methods for the computation of ground state properties \cite{marzari2012maximally}, as well as the modern theory of polarization \cite{resta1992theory, king1993theory, resta1993macroscopic}.

The notion of well-localized bases for the Fermi projection was introduced by Wannier in \cite{wannier1937structure}, where such bases were constructed in the periodic setting by integrating Bloch functions over the Brillouin zone. Since then a large number of works have appeared on the existence of well-localized Wannier functions in the periodic case \cite{kohn1959analytic, descloizeaux1964energy, marzari1997maximally, brouder2007exponential, kivelson1982wannier}. It emerged that the existence of well-localized Wannier bases is closely tied to topological triviality of the Bloch bundle \cite{monaco2018optimal}. In particular, a translation invariant gapped Fermi projector in two dimensions admits an exponentially localized Wannier basis if and only if its \emph{Chern number} vanishes. This result was dubbed the \emph{localization dichotomy for periodic insulators} in \cite{marcelli2019haldane}.

In the same paper \cite[~Sect. 5]{marcelli2019haldane} the authors introduced the \emph{localization dichotomy conjecture} for two-dimensional non-periodic insulators, which states that a gapped Fermi projection admits an exponentially localized Wannier basis if and only if its Chern number vanishes, also in the non-periodic case. In this paper we prove this localization dichotomy conjecture for non-periodic insulators on the lattice $\Z^2$. The main novelty of this work consists in proving the implication 
$$ 
\Ch(P) = 0 \quad \implies \quad P \, \text{admits an exponentially localized Wannier basis.}
$$
The converse was already shown for two-dimensional insulators in the continuum in \cite{marcelli2023localization}. We also provide a particularly simple proof of this implication in the lattice setting, see Theorem \ref{thm:rapidly local Wannier basis implies Ch(P) = 0}.

\subsection{Strategy of proof}

In order to prove that an exponentially local projector $P$ on $l^2(\Z^2) \otimes \C^n$ admits an exponentially localized Wannier basis if $\Ch(P) = 0$ we will decompose
$$
    P = \sum_{j \in \Z} P_j
$$
as a sum over an orthogonal family of rapidly local projections $P_j$, each of which is localized near the horizontal line $\Hyper_j$ at height $jL$, for some length scale $L$. Rapid locality means off-diagonal decay that is faster than any inverse power in the distance. The class of rapidly local operators is stable under the smooth functional calculus (Proposition \ref{prop:rapid locality is preserved by smooth functions}).

The projections $P_j$ describe quasi one-dimensional systems, which can be shown to admit rapidly localized Wannier bases by the methods of \cite{nenciu1993existence, nenciu1998existence}. We provide an alternative proof of this fact adapted to our setup (Theorem \ref{thm:1D rapidly local Wannier}).

The crucial difficulty therefore lies in obtaining the decomposition $P = \sum_j P_j$. To achieve this we construct half-space projectors $F_j$ whose difference with the compressed half-space projectors $P \Pi_j P$ is rapidly localized near the cut $\Hyper_j$, where $\Pi_j$ projects onto the upper half-plane above $\Hyper_j$ (Proposition \ref{prop:reducing half-space projectors}). The differences $D_j = F_{j} - F_{j+1}$ are then approximate projectors which are approximately orthogonal to each other, and sum to $P$. In a final step we use functional calculus to turn the $D_j$ into actual projections, and then rotate them slightly to form the required family $\{ P_j \}_{j \in \Z}$ (Proposition \ref{prop:exact strip decomposition}).

The key new idea is in the construction of the half-space projections $F_j$. The compressions $P \Pi_j P$ are generally not projections, in fact the bulk-boundary correspondence \cite{kellendonk2002edge,elbau2002equality,prodan2016bulk, bols2022absolutely} implies that if $\Ch(P) \neq 0$ then $\sigma(P \Pi_j P) = [0, 1]$ fills the whole interval. Note however that if the spectrum of $P \Pi_j P$ were to have a gap in $(0, 1)$, uniformly in $j \in \Z$, then we could use the smooth functional calculus to push all the spectrum above the gap to 1, and all the spectrum below the gap to 0, obtaining a projector $F_j$ which inherits the locality of $P \Pi_j P$. This sort of condition is very similar to the uniform gaps assumption in the compressed position operator made in \cite{lu2022existence}. By the bulk-boundary correspondence, such gaps cannot exist if $\Ch(P) \neq 0$. However, even if $\Ch(P) = 0$ such a gap is not guaranteed. Indeed, the bulk boundary correspondence for topological insulators protected by time-reversal symmetry implies that $\sigma(P \Pi_j P)$ also fills the whole interval $[0, 1]$ if the time-reversal symmetric $P$ has a non-trivial Fu-Kane-Mele invariant \cite{graf2013bulk, fonseca2020two, bols2023fredholm, bols2024absolutely}, while $\Ch(P) = 0$ for such time-reversal invariant insulators.

In order to overcome this difficulty, we must show that if $\Ch(P)=0$ then the $P \Pi_j P$ can be perturbed in a locality-respecting way to open a uniform gap. This is achieved as follows. We define the \emph{boundary unitaries}
$$
U_j := \ed^{2 \pi \iu P \Pi_j P}.
$$
These unitaries have $U_j - \I$ rapidly local and localized near the horizontal cut $\Hyper_j$. In other words, they are quasi-one dimensional unitaries, and we can associate an integer \emph{flow} to them (for example, the flow of the bilateral shift on $l^2(\Z)$ equals one). One version of the bulk-boundary correspondence \cite{bols2023fredholm} then states that
$$
\Flow(U_j) = \Ch(P).
$$
Under the hypothesis that $\Ch(P) = 0$, we therefore have that all boundary unitaries $U_j$ have zero flow.

Using Lin's theorem
\cite{lin1993almostcommuting,friis1996almost,hastings2009making,herrera2020hastings} on almost commuting matrices, we then prove that one-dimensional rapidly local unitaries with zero flow can be perturbed in a locality respecting way to open up a gap around -1 (Proposition \ref{prop:gapping 1D unitaries}). Applying this to the boundary unitaries produces rapidly local unitaries $V_j$ which are norm-close to $U_j$ and have a uniform gap around -1 (Lemma \ref{lem:gapping out boundary unitaries}). This construction is followed by a spectral gluing construction, producing rapidly local unitaries $W_j$ which are norm-close to $U_j$, retain the gap around -1, and agree with $U_j$ on a spectral interval around 1 (Lemma \ref{lem:spectral gluing}). The latter condition allows us to `invert the exponential' to obtain operators $G_j$ which satisfy $W_j = \ed^{2 \pi \iu G_j}$ and such that the difference $G_j - P \Pi_j P$ is rapidly localized near $\Hyper_j$, leading to a proof of Proposition \ref{prop:reducing half-space projectors}. (It is the latter property which is guaranteed by the condition that $W_j$ agree with $U_j$ on a spectral interval around 1). Crucially, $G_j$ now inherits a gap at 1/2 from $W_j$'s gap at -1. We can now use the smooth functional calculus to turn the $G_j$ into the projectors $F_j$. We thereby obtain a rapidly localized Wannier basis for $P$ (Theorem \ref{thm:2D rapidly local Wannier}).

In the final step, we use exponential locality of $P$ to upgrade this rapidly local Wannier basis to an exponentially localized one (Theorem \ref{thm:rapid to exponential}).

\subsection{Relation to prior work}

A position-space approach to the construction of localized Wannier bases was already proposed by Kivelson, who used
eigenstates of a projected position operator to define Wannier functions
in one-dimensional disordered systems \cite{kivelson1982wannier}.

This method was further developed by Nenciu and Nenciu
\cite{nenciu1993existence,nenciu1998existence} who proved existence of exponentially localized Wannier functions for gapped systems in one dimension in the non-periodic setting. They show that the compressed position operator $PXP$ has discrete spectrum, and conclude that the corresponding eigenstates form the required Wannier basis.

To get a result in the two-dimensional case Lu, Stubbs, and Watson followed a similar approach in \cite{lu2022existence}. Under the assumption that the compressed position operator $PXP$ has uniform spectral gaps, they show that the spectral projections associated to the spectral patches of $P X P$ are local projections localized near lines of constant $X$-coordinate. These projections can in turn be shown to admit localized Wannier bases using the methods of \cite{nenciu1993existence, nenciu1998existence}. This work is a clear precedent for the strategy followed in this paper, to the extent that they also first chop up the Fermi projection into strips. However, the method of obtaining this strip decomposition is significantly different, the present paper does not rely on any assumptions on the exponentially local projector $P$ other than the vanishing of its Chern number.

The improvement of an already localized basis to an exponentially
localized one also has a non-periodic precedent. Lu and Stubbs proved
that sufficiently strong algebraic localization implies exponential
localization in two dimensions \cite{lu2021algebraic}. Theorem~\ref{thm:rapid to exponential} provides
a direct lattice argument, valid in arbitrary dimension, starting from
rapid localization.

\paragraph{Organization of the paper.}
Section~\ref{sec:setup and main results} gives the essential definitions and the main theorem. Section~\ref{sec:locality and functional calculus} develops
the locality estimates and functional calculus techniques used throughout the proof.
Section~\ref{sec:boundary unitaries and the Chern number} introduces the boundary unitaries and their flow indices, and relates these invariants to the bulk Chern number.
Section~\ref{sec:gapping 1D unitaries} proves the gapping result for 1D unitaries with trivial flow index. Section~\ref{sec:rapidly localized Wannier bases in 2D} constructs
the strip decomposition and rapidly localized Wannier bases for Chern trivial rapidly local projectors. Section~\ref{sec:rapid to exponential}
upgrades rapid localization to exponential localization, and Section~\ref{sec:rapid Wannier implies Ch(P) = 0}
proves that a rapidly localized Wannier basis implies Chern triviality.

\paragraph{Usage of LLMs.} This work was produced with the help of ChatGPT 5.6. Most significantly, the proof strategy for Proposition \ref{prop:gapping 1D unitaries} was generated by the LLM. This strategy applies Lin's theorem on almost commuting matrices \cite{lin1993almostcommuting}, and otherwise uses techniques introduced in the related literature \cite{hastings2009making, herrera2020hastings}. In addition, through interacting with the LLM the author learned the useful trick of associating to an almost-orthogonal set $\{w_{i} \}_{i \in \caI} \subset \caH$ a map $\Phi : l^2(\caI) \to \caH : e_i \mapsto w_i$ which is norm-close to an isometry, and which can then be deformed to an isometry using the functional calculus. This technique goes back to \cite{lowdin1950non}, see also \cite{frank2002symmetric, grochenig2004localization, benzi2013decay}. Variations of this trick are used many times throughout the manuscript. 

\paragraph{Acknowledgments.} The author would like to thank Simone Warzel for helpful discussions, and for sharing her LLM generated proof of a Hamiltonian version of Proposition \ref{prop:gapping 1D unitaries}.

\section{Setup and main results} \label{sec:setup and main results}

Consider the Hilbert space $\H{d}{n} = l^2(\Z^d) \otimes \C^n$. Let $\{ \delta_x \}_{x \in \Z^d}$ denote the position basis of $l^2(\Z^d)$ and let $\{ e_{i} \}_{i = 1}^n$ be the canonical basis of $\C^n$. Then $\delta_{x, i} := \delta_x \otimes e_i$ are the elements of an orthonormal basis of $\H{d}{n}$. For any $x \in \Z^d$, let $\iota_x : \C^n \to \H{d}{n}$ be the injection defined by $\iota_x e_i = \delta_{x, i}$. For any $x, y \in \Z^d$ and any $A \in \caB(\H{d}{n})$ we write
$$ A(x, y) := \iota^*_x A \iota_y $$
for the matrix-valued kernel of $A$. We write $\abs{A(x, y)}$ for the spectral norm of $A(x, y)$.

We consider $\Z^d$ as a subset of $\R^d$, and equip both $\Z^d$ and $\R^d$ with the Euclidean metric which we denote by $\abs{x - y} = \sqrt{ (x_1 - y_1)^2 + \cdots + (x_d - y_d)^2 }$. An operator $A \in \caB(\H{d}{n})$ is called \emph{exponentially local} if there are constants $C < \infty$ and $\mu > 0$ such that
$$ \abs{ A(x, y) } \leq C \ed^{- \mu \abs{x-y}} \quad \text{for all} \quad  x, y \in \Z^d. $$

\subsection{Chern number} \label{subsec:Chern number}

For any projector $P$ and unitary $U$ acting on $\H{d}{n}$ for which $U P U^* - P$ is compact we define an index \cite{avron1994index}
\begin{equation} \label{eq:index of P and U}
    \ind(P, U) := \dim \ker ( U P U^* - P - \I ) - \dim \ker ( U P U^* - P + \I ) \in \Z.
\end{equation}
This index is stable under compact perturbation and norm-continuous deformations of both $P$ and $U$, as long as $UPU^* - P$ remains compact \cite{avron1994index}.

We now specialize to $d = 2$, so we consider $\H{2}{n} = l^2(\Z^2) \otimes \C^n$. Let $P = P^* = P^2$ be an exponentially local projector on $\H{2}{n}$ and let
$$F_a = \exp \left( \iu \arg( (X_1 - a_1) + \iu (X_2 - a_2) ) \right)$$
be the flux insertion operator at $a \in \R^2 \setminus \Z^2$. Then $F_a^* P F_a - P$ is compact (\cite[~Sect. 4]{aizenman1998localization}) and we can define the \emph{Chern number} of $P$ by
$$ \Ch(P) := \ind(P, F_a) $$
which is independent of $a$.

\subsection{Localized Wannier bases}

For a vector $\psi \in \H{d}{n}$ we write $\psi(x) = \iota^*_x \psi \in \C^n$ for its component at $x \in \Z^d$. We write $\abs{\psi(x)}$ for the norm of this component.

The following definition is an adaptation to the lattice setting of the notion of `generalized Wannier basis' introduced in the continuum setting by \cite{nenciu1993existence, nenciu1998existence}.

\begin{definition} \label{def:exponentially localized Wannier}
    An orthogonal projector $P \in \caB(\H{d}{n})$ admits an exponentially localized Wannier basis if there exist constants $C < \infty$ and $\mu > 0$, an orthonormal basis $\{ w_{\al} \}_{\al \in \caA}$ of $\Ran P$, and localization centres $x_{\al} \in \R^d, \al \in \caA$ such that
    $$ \abs{ w_{\al}( x ) } \leq C \ed^{-\mu \abs{ x - x_{\al} }} \quad \text{for all} \quad x \in \Z^d, \al \in \caA.$$ 
\end{definition}

\begin{remark}
    In many continuum versions of the notion of a generalized Wannier basis some notion of non-clustering or uniformity for the localization centers is required as part of the definition \cite{marcelli2019haldane, marcelli2023localization}. It was shown in \cite[Lemma 2.1.]{lu2024algebraic} that such non-clustering properties follow from the other assumptions. We provide the analogue in the lattice setting in Lemma \ref{lem:no clustering}.
\end{remark}

\subsection{Main Result}

The following theorem is an immediate consequence of Theorem \ref{thm:2D rapidly local Wannier}, Theorem \ref{thm:rapid to exponential} and Theorem \ref{thm:rapidly local Wannier basis implies Ch(P) = 0}.

\begin{theorem} \label{thm:Localization dichotomy}
    Let $P$ be an exponentially local projection on $\H{2}{n}$. Then $P$ admits an exponentially localized Wannier basis if and only if $\Ch(P) = 0$.
\end{theorem}

\section{Locality and functional calculus} \label{sec:locality and functional calculus}

Even though our initial projector $P$ is exponentially local, our constructions below will often involve taking smooth functions of various local operators, and exponential locality is not stable under such operations. The good notion of locality that is stable under the smooth functional calculus is \emph{rapid locality}, meaning that off diagonal matrix elements decay faster than any polynomial.

Many basic results and techniques presented in this section related to algebraic and ideal properties of appropriately local and localized operators go back to \cite{jaffard1990proprietes}, see also \cite{grochenig2010noncommutative}.

\subsection{Rapid locality}

For any $v \in \R^d$ write $ \jb{v} := 1 + \abs{v}$. Note that
\begin{equation} \label{eq:triangle for the Japanese bracket}
    \jb{a + b} \leq \jb{a} \jb{b}
\end{equation}
for all $a, b \in \R^d$.

\begin{definition} \label{def:rapid locality}
    An operator $A \in \caB(\H{d}{n})$ is rapidly local if for each $N \in \N$ there is a constant $C_N < \infty$ so that
    $$ \abs{A(x, y)} \leq C_N \jb{x-y}^{-N} \quad \text{for all} \quad x, y \in \Z^d. $$
    We say that $A$ is rapidly local with constants $\{C_N \}_{N \in \N}$. We write $\LOC(\H{d}{n})$ for the set of rapidly local operators on $\H{d}{n}$.

    More generally, given a subset $\Sigma \subset \R^d$, we say $A$ is rapidly local and localized near $\Sigma$ with constants $\{ C_N \}_{N \in \N}$ if
    $$ \abs{A(x, y)} \leq C_N \jb{x-y}^{-N} \jb{ \dist( \Sigma, x ) + \dist( \Sigma, y) }^{-N} \quad \text{for all} \quad x, y \in \Z^d, N \in \N. $$
    We write $\LOC_{\Sigma}(\H{d}{n})$ for the set of operators which are rapidly local and localized near $\Sigma$. Note that $\LOC( \H{d}{n} ) = \LOC_{\R^d}( \H{d}{n} )$.
    
    A family $\{A_i\}_{i \in \caI}$ of operators on $\H{d}{n}$ is uniformly rapidly local and localized near subsets $\{\Sigma_{i} \subset \R^d\}_{i \in \caI}$ if each $A_i$ is rapidly local and localized near $\Sigma_i$ with constants $\{C_N\}_{N \in \N}$ independent of $i \in \caI$. 
    
    We say the family $\{A_i\}_{i \in \caI}$ is uniformly rapidly local and localized near $\Sigma \subset \R^d$ if it is uniformly rapidly local and localized near $\{ \Sigma_i \}_{i \in \caI}$ with $\Sigma_i = \Sigma$ for all $i \in \caI$. We say the family $\{A_i\}_{i \in \caI}$ is uniformly rapidly local if it is uniformly rapidly local and localized near $\R^d$.
\end{definition}

It will be convenient to characterize rapid locality in terms of weighted Schur norms. Define for any $N \in \N, \Sigma \subset \R^d$ the norm
\begin{align*} 
    \norm{A}_{N, \Sigma} := \max \bigg\lbrace \sup_{x} \sum_{y} &\abs{A(x, y)} \jb{x - y}^N \jb{ \d(\Sigma, x) + \d(\Sigma, y) }^N, \\ &\sup_y \sum_x \abs{A(x, y)} \jb{x - y}^N \jb{ \d(\Sigma, x) + \d(\Sigma, y) }^N  \bigg\rbrace
\end{align*}
where $\d(\Sigma, x) = \inf_{y \in \Sigma} \abs{x - y}$ denotes the Euclidean distance between $\Sigma$ and $x$. Write also $\norm{A}_N := \norm{A}_{N, \R^d}$.

\begin{lemma} \label{lem:norm characterization of rapid locality}
    If an operator $A \in \caB(\H{d}{n})$ has $\norm{A}_{N, \Sigma} < \infty$ for all $N \in \N$ then $A$ is rapidly local and localized near $\Sigma$ with constants $\{ C_N = \norm{A}_{N, \Sigma} \}_{N \in \N}$.
    
    If an operator $A \in \caB(\H{d}{n})$ is rapidly local and localized near $\Sigma \subset \R^d$ with constants $\{ C_N \}_{N \in \N}$, then $\norm{A}_{N, \Sigma} \leq C_{N + d + 1} c_d$ for all $N \in \N$, where $c_d = \sup_x \sum_y \jb{x-y}^{-(d+1)} < \infty$.
    
    In particular, a family $\{A_i \}_{i \in \caI}$ of bounded operators on $\H{d}{n}$ is uniformly rapidly local and localized near $\{ \Sigma_i \}_{i \in \caI}$ if and only if there are constants $\{ D_N \}_{N \in \N}$ such that $\norm{A_i}_{N, \Sigma_i} \leq D_N$ for all $N \in \N$.
\end{lemma}

\begin{proof}
    If $\norm{A}_{N, \Sigma} < \infty$ for all $N \in \N$ then we immediately obtain
    $$ \abs{A(x, y)} \jb{x-y}^{N} \jb{\d(\Sigma, x) + \d(\Sigma, y)}^{N} \leq \norm{A}_{N, \Sigma} $$
    for all $N \in \N$, so $A$ is rapidly local and localized near $\Sigma$ with constants $C_N = \norm{A}_{N, \Sigma}$.

    Conversely, suppose $A$ is rapidly local and localized near $\Sigma$ with constants $\{C_N\}_{N \in \N}$. Then
    \begin{align*}
        \sum_y &\abs{A(x, y)} \jb{x-y}^N \jb{ \d(\Sigma, x) + \d(\Sigma, y) }^N \\
        &\leq  C_{N + d + 1} \sum_y \jb{x-y}^{-(d+1)} \jb{ d(\Sigma, x) + \d(\Sigma, y) }^{-(d+1)} \\
        &\leq C_{N + d + 1} c_d.
    \end{align*}
    This bound is independent of $x$, and exchanging the roles of $x$ and $y$ gives an identical bound. We conclude that $\norm{A}_{N, \Sigma} \leq C_{N + d + 1} c_d$ for each $N \in \N$. 
\end{proof}

\begin{lemma} \label{lem:*-algebras and ideals of rapidly local operators}
    We have for all $\Sigma \subset \R^d, N \in \N$ and all $A, B \in \caB(\H{d}{n})$ that
    \begin{enumerate}
        \item $ \norm{A} \leq \norm{A}_{0, \Sigma}$, \label{eq:operator norm bounded by Schur norm}
        \item $\norm{A^*}_{N, \Sigma} = \norm{A}_{N, \Sigma}$, \label{eq:Schur norm is invariant under adjoints}
        \item $\norm{AB}_{N, \Sigma} \leq \norm{A}_{N, \Sigma} \norm{B}_{N, \Sigma}$, \label{eq:sub-multiplicativity}
        \item $\norm{A}_{N} \leq \norm{A}_{N, \Sigma}$, \label{eq:local and localized is stricter than local}
        \item $\norm{A B}_{N, \Sigma} \leq \norm{A}_{2N} \norm{B}_{N, \Sigma}$ and $\norm{A B}_{N, \Sigma} \leq \norm{A}_{N, \Sigma} \norm{B}_{2N}$. \label{eq:Schur norm property for ideal}
    \end{enumerate}
    In particular, we have that $\LOC_{\Sigma}(\H{d}{n})$ is a *-subalgebra of $\caB(\H{d}{n})$ and that $\LOC_{\Sigma}(\H{d}{n})$ is a two-sided ideal in $\LOC(\H{d}{n})$.
\end{lemma}

\begin{proof}
    \eqref{eq:operator norm bounded by Schur norm}. For any finitely supported $\psi \in \H{d}{n}$ we have
    \begin{align*}
        \norm{ A \psi }^2 &\leq \sum_{x} \left( \sum_y \abs{A(x, y)} \abs{\psi(y)} \right)^2 \\
        &\leq \sum_x \left( \sum_y  \abs{A(x, y)}\right) \left( \sum_y \abs{A(x, y)} \abs{\psi(y)}^2 \right) \\
        &\leq \norm{A}_0 \, \sum_y \left( \sum_x \abs{A(x, y)} \right) \abs{\psi(y)}^2 \\
        &\leq \norm{A}_0^2 \norm{\psi}^2,
    \end{align*}
    where we used Cauchy-Schwarz to obtain the second line. This proves the first claim of the Lemma.

    \eqref{eq:Schur norm is invariant under adjoints}. Follows immediately from the definition.
    
    \eqref{eq:sub-multiplicativity}. Use Eq. \eqref{eq:triangle for the Japanese bracket} to bound
    \begin{align*}
         \sum_y \, &\abs{ (AB)(x, y) } \jb{x-y}^N \jb{\d(\Sigma, x) + \d(\Sigma, y)}^N  \\
        &\leq \sum_{y, z} \, \abs{ A(x, z) } \abs{B(z, y)} \, \jb{x-z}^N \jb{z-y}^N \jb{\d(\Sigma, x) + \d(\Sigma, z)}^N \jb{\d(\Sigma, z) + \d(\Sigma, y)}^N \\
        &= \sum_z \abs{A(x, z)} \jb{x-z}^N \jb{\d(\Sigma, x) + \d(\Sigma, z)}^N \, \left( \sum_y \abs{B(z, y)} \jb{z - y}^N \jb{\d(\Sigma, z) + \d(\Sigma, y)}^N \right) \\
        &\leq \norm{B}_{N, \Sigma} \, \sum_z \abs{A(x, z)} \jb{x-z}^N \jb{\d(\Sigma, x) + \d(\Sigma, z)}^N \\
        &\leq \norm{A}_{N, \Sigma} \norm{B}_{N, \Sigma}
    \end{align*}
    independently of $x$. There is a similar bound when we sum over $x$ at fixed $y$. This proves sub-multiplicativity of $\norm{\cdot}_{N, \Sigma}$.

    \eqref{eq:local and localized is stricter than local}. Follows immediately from the definition.

    \eqref{eq:Schur norm property for ideal}. First note that
    $$ \abs{ \d(\Sigma, x) - \d(\Sigma, y) } \leq \abs{x-y}. $$
    Indeed, suppose that $\d(\Sigma, x) \geq \d(\Sigma, y)$, and that $\d(\Sigma, y) = \abs{z - y}$ for $z \in \overline \Sigma$ in the closure of $\Sigma$. Then
    $$ \d(\Sigma, x) \leq \abs{z - x} \leq \abs{z - y} + \abs{x - y} = \d(\Sigma, y) + \abs{x - y}, $$
    establishing the inequality in this case. The case $\d(\Sigma, x) \leq \d(\Sigma, y)$ is dealt with in the same way. Using this inequality we obtain
    $$ \jb{\d(\Sigma, x) + \d(\Sigma, y)} \leq \jb{x - z} \jb{\d(\Sigma, z) + \d(\Sigma, y)}.$$
    We can now use this inequality in a similar calculation to the one we used to prove item \eqref{eq:sub-multiplicativity} to obtain
    \begin{align*}
        \sum_y \, &\abs{ (AB)(x, y) } \jb{x-y}^N \jb{\d(\Sigma, x) + \d(\Sigma, y)}^N  \\
        &\leq \sum_{y, z} \, \abs{ A(x, z) } \abs{B(z, y)} \, \jb{x-z}^{2N} \jb{z-y}^N \jb{\d(\Sigma, z) + \d(\Sigma, y)}^N. \\
        &= \sum_z \abs{A(x, z)} \jb{x-z}^{2N} \, \left( \sum_y \abs{B(z, y)} \jb{z - y}^N \jb{\d(\Sigma, z) + \d(\Sigma, y)}^N \right) \\
        &\leq \norm{A}_{2N} \norm{B}_{N, \Sigma}
    \end{align*}
    uniformly in $x$. A similar bound holds when we sum over $x$ at fixed $y$, establishing $\norm{AB}_{N, \Sigma} \leq \norm{A}_{2N} \norm{B}_{N, \Sigma}$. The bound $\norm{AB}_{N, \Sigma} \leq \norm{A}_{N, \Sigma} \norm{B}_{2N}$ can either be obtained in the same way, or from the inequality already established and the invariance of the weighted Schur norms under taking adjoints. 
    
    The characterization of $\LOC_{\Sigma}(\H{d}{n})$ provided by Lemma \ref{lem:norm characterization of rapid locality}, the sub-multiplicativity \eqref{eq:sub-multiplicativity}, and invariance under adjoints \eqref{eq:Schur norm is invariant under adjoints} imply that $\LOC_{\Sigma}(\H{d}{n})$ is a *-algebra. Finally, items \eqref{eq:local and localized is stricter than local} and \eqref{eq:Schur norm property for ideal} imply that $\LOC_{\Sigma}(\H{d}{n})$ is a two-sided ideal in $\LOC(\H{d}{n})$.
\end{proof}

\subsection{Rapidly local Combes-Thomas estimates}

Combes-Thomas estimates \cite{combes1973asymptotic} control the off-diagonal decay of the resolvents of a given operator $A$ in terms of the off-diagonal decay of $A$ itself. The versions available in the literature are not suitable for our needs. Here we prove Combes-Thomas estimates for rapidly local self-adjoint and unitary operators on $\H{d}{n}$.

\begin{proposition} \label{prop:rapidly local Combes-Thomas}
    Let $A$ be a normal operator on $\H{d}{n}$ which is rapidly local with constants $\{ C_N \}_{N \in \N}$. Then there are constants $\{ D_N \}_{N \in \N}$ depending only on $\{C_N\}_{N \in \N}$ and $d$, such that for any $z \in \C \setminus \sigma(A)$ and any $N \in \N$ we have
    $$ \abs{ (z - A)^{-1}(x, y) } \leq D_N \min \{ \eta, 1 \}^{-(N+1)} \jb{x-y}^{-N} \quad \text{for all } x , y \in \Z^d, $$
    where $\eta = \d\big( z, \sigma(A) \big)$ is the distance from $z$ to the spectrum of $A$. 
\end{proposition}

\begin{remark}
    The proof actually goes through for any bounded operator $A \in \caB(\H{d}{n})$ and any $z \in \C$ for which the resolvent $(z - A)^{-1}$ exists as a bounded operator, and taking $\eta = \norm{(z - A)^{-1}}^{-1}$.
\end{remark}

\begin{proof}
    Fix an integer $N \geq 1$, a lattice site $a \in \Z^d$, a cut-off $R > 0$, and a number $0 < \ep \leq 1$, and define a multiplication operator $M \in \caB(\H{d}{n})$ by
    $$ M \delta_{x, i}  = w(x) \delta_{x, i}, \quad w(x) = (1 + \ep \min \{ \abs{x - a}, R \})^{-N} $$
    for all $x \in \Z^d$ and all $i \in \{1, \cdots, n\}$. Since $r(x) := \min\{ \abs{x - a}, R \}$ satisfies $r(x) \leq r(y) + \abs{x - y}$ for any $x, y \in \Z^d$ we have
    $$ 1 + \ep r(x) \leq 1 + \ep r(y) + \ep \abs{x-y} \leq (1 + \ep r(y))(1 + \ep\abs{x-y}) $$
    and therefore
    $$ \frac{w(x)}{w(y)} \leq (1 + \ep \abs{x-y})^N. $$
    Since this also holds with $x$ and $y$ exchanged, we get
    $$ \abs{ \frac{w(x)}{w(y)} - 1 } \leq (1 + \ep \abs{x - y})^N - 1 \leq \ep \jb{x - y}^N.$$
    Indeed, if $w(x)/w(y) \geq 1$ the first inequality follows immediately from the above. Otherwise
    $$ \abs{ \frac{w(x)}{w(y)} - 1 } = 1 - \frac{w(x)}{w(y)} = \frac{w(x)}{w(y)} \left( \frac{w(y)}{w(x)} - 1 \right) $$
    where now the first factor is less than one, and the second is bounded as before.

    Using this bound we can now establish
    \begin{equation} \label{eq:bound on MAM-A}
        \norm{MAM^{-1} - A} \leq \ep \norm{A}_N.
    \end{equation}
    Indeed, for any compactly supported $\psi \in \H{d}{n}$ we have
    \begin{align*}
        \norm{( MAM^{-1} - A )\psi}^2 &\leq \sum_x \, \left\str  \sum_y \, (M A M^{-1} - A)(x, y) \, \psi(y) \right\str^2 \\
        &=\sum_x \, \left\str  \sum_y \,  \left( \frac{w(x)}{w(y)} - 1 \right) \, A(x, y) \, \psi(y) \right\str^2 \\
        &\leq \sum_x \left(  \sum_y \, \abs{A(x, y)} \, \ep \jb{x-y}^N \right) \left( \sum_y \, \abs{A(x, y)} \, \ep \jb{x-y}^N \abs{\psi(y)}^2 \right) \\
        &\leq \ep^2 \norm{A}_N^2 \norm{\psi}^2
    \end{align*}
    where we used Cauchy-Schwarz to obtain the third line.

    Set
    $$ a_N := \frac{1}{2(1 + \norm{A}_N)}, \quad \ep := a_N \min\{\eta, 1\}. $$
    Writing $B = M A M^{-1} - A$, the bound \eqref{eq:bound on MAM-A} and the fact that $\norm{(z-A)^{-1}} \leq \eta^{-1}$ then implies
    $$ \norm{ B (z - A)^{-1} } \leq \frac{1}{2}, $$
    and inverting
    $$ M (z - A) M^{-1} = z - M A M^{-1} = \big( \I - B (z - A)^{-1} \big) \big( z - A \big) $$
    gives us
    $$ M (z - A)^{-1} M^{-1} = (z - A)^{-1} \big( \I - B(z - A)^{-1} \big)^{-1}. $$
    and therefore
    \begin{equation} \label{eq:bound on MRM}
        \norm{ M (z - A)^{-1} M^{-1} } \leq \frac{2}{\eta}.
    \end{equation}

    By choosing $a = x$, so that $w(x) = 1$, and taking $R > \abs{x-y}$ we find
    $$ \abs{ ( M(z - A)^{-1}M^{-1} )(x, y) } = (1 + \ep \abs{x-y})^N \abs{ (z - A)^{-1}(x, y) } $$
    and therefore, using Eq. \eqref{eq:bound on MRM} and recalling that we set $\ep = a_N \min \{ \eta, 1 \}$,
    \begin{equation} \label{eq:CT estimate}
        \abs{ (z - A)^{-1}(x, y) } \leq \frac{2}{\eta} \, \big( 1 + a_N \min\{\eta, 1\} \abs{x-y} \big)^{-N}.
    \end{equation}

    If $\eta \leq 1$ then $1 + a_N \eta \abs{x-y} \geq a_N \eta (1 + \abs{x-y})$ so
    $$ \abs{ (z-A)^{-1}(x, y) } \leq 2 a_N^{-N} \eta^{-(N+1)} \jb{x-y}^{-N}. $$
   
    If $\eta > 1$ then similar reasoning gives
    $$ \abs{(z-A)^{-1}(x, y)} \leq 2 a_N^{-N} \jb{x-y}^{-N}. $$
    Lemma \ref{lem:norm characterization of rapid locality} gives upper bounds on $\norm{A}_N$, and therefore lower bounds on $a_N$, in terms of the constant $C_{N + d + 1}$ and the dimension $d$. We can therefore upper bound the constants $2 a_N^{-N} \leq D_N < \infty$ by constants $D_N$ that depend only on the $\{ C_N \}_{N \in \N}$ and on $d$. This proves the claim.
\end{proof}

\subsection{Smooth functions of rapidly local operators}

In this subsection we will prove that smooth functions of rapidly local self-adjoint and unitary operators are still rapidly local and localized. The crucial tools are Helffer-Sj\"{o}strand formulas for the smooth functional calculus, and the Combes-Thomas estimates proved in the previous subsection.

\subsubsection{Helffer-Sj\"{o}strand formulas}

The Helffer-Sj\"{o}strand formula for smooth functions of self-adjoint operators goes back to \cite{helffer2005equation}. See also \cite[Sec. 2.2.]{davies1995spectral} and \cite[Sec. 9]{hunziker2000quantum}, among others. We review this formula now.

Given a compactly supported smooth function $f : \R \to \C$ and any $N \geq 2$ we may consider the almost analytic extension $\tilde f : \C \to \C$ of $f$ given by
$$ \tilde f_N := \chi(z) \sum_{k = 0}^N f^{(k)}(x) \frac{(\iu y)^k}{k!},  $$
where we wrote $z = x + \iu y$ and where $\chi \in C^{\infty}(\C)$ is supported on the strip $\{ x + \iu y \in \C \, : \, \abs{y} < 1 \}$, and $\chi(z) = 1$ for all $z$ in a neighbourhood of $\supp(f)$ in the complex plane. It is an extension of $f$ in the sense that $\tilde f_N(x) = f(x)$ for all $x \in \R$, and it is almost analytic in the sense that $\abs{ \partial_{\bar z} \tilde f_N(z) } \leq c_N \abs{y}^{N}$ for some constants $c_N < \infty$. Moreover, $\partial_{\bar z} \tilde f_N$ is compactly supported on $\{x + \iu y \in \C \, : \, x \in \supp(f), \abs{y} < 1 \}$. Here, $\partial_{\bar z} = \frac{1}{2}(\partial_x + \iu \partial_y)$, following the conventions of \cite{davies1995spectral}. Moreover, $\tilde f_N$ and all its derivatives are supported in the strip $\{ x + \iu y \in \C \, : \, \abs{y} < 1 \}$. The Helffer-Sj\"{o}strand formula then says that for any bounded self-adjoint operator $A$ on a separable Hilbert space $\caH$ we have
\begin{equation} \label{eq:HS for self-adjoints}
    f(A) = -\frac{1}{\pi} \, \int_{\C} \dd x \dd y \, \partial_{\bar z} \tilde f_N(z) \, (z - A)^{-1}.
\end{equation}

We will also need a Helffer-Sj\"{o}strand formula for smooth functions of a unitary operator $U$ on $\caH$. The version given in \cite[~Lemma A.6.]{shapiro2019strongly} is especially well-suited for our purposes. Let us review it here.

For any smooth function $g \in C^{\infty}(\T)$ there exists for any $N \geq 2$ an almost analytic extension $\tilde g_N \in C^{\infty}(\C)$ of $g$ such that $\abs{ \partial_{\bar z} \tilde g_N(z) } \leq c_N \abs{\abs{z} - 1}^{N}$ for a constant $c_N < \infty$. Moreover, $\tilde g_N$ can be chosen so that it is supported on the annulus $\{ z \in \C \,: \, 1/2 < \abs{z} < 1+1/2 \}$. For any such $\tilde g_N$ we then have
\begin{equation} \label{eq:HS for unitaries}
    g(U) = -\frac{1}{\pi} \int_{\C} \dd x \dd y \, \partial_{\bar z} \tilde g_N(z)  \, (z - U)^{-1}.
\end{equation}
(The apparent difference with Eq. A.2. of \cite[~Lemma A.6.]{shapiro2019strongly} is because $\dd z \dd \bar z = - 2 \iu \dd x \dd y$.)

\subsubsection{Smooth functions preserve rapid locality and localization}

\begin{proposition} \label{prop:rapid locality is preserved by smooth functions}
    Let $\{ A_i \}_{i \in \caI}$ be a family of self-adjoint bounded operators on $\H{d}{n}$ which is uniformly rapidly local, and let $f \in C^{\infty}_c(\R)$ be a compactly supported smooth function. Then the family $\{ f(A_i) \}_{i \in \caI}$ is uniformly rapidly local.

    If $\{ B_i \}_{i \in \caI}$ is another family of self-adjoint bounded operators on $\H{d}{n}$ which is uniformly rapidly local, and such that $\{ A_i - B_i \}_{i \in \caI}$ is uniformly rapidly local and localized near some $\{ \Sigma_i \}_{i \in \caI}$, then the family $\{ f(A_i) - f(B_i) \}_{i \in \caI}$ is uniformly rapidly local and localized near $\{ \Sigma_i\}_{i \in \caI}$.

    Similarly, let $\{ U_i \}_{i \in \caI}$ be a family of unitary operators on $\H{d}{n}$ which is uniformly rapidly local, and let $g \in C^{\infty}(\T)$ be a smooth function on the unit circle. Then the family $\{ g(U_i) \}_{i \in \caI}$ is uniformly rapidly local.

    If $\{ V_i \}_{i \in \caI}$ is another family of unitary operators on $\H{d}{n}$ which is uniformly rapidly local, and such that $\{ U_i - V_i \}_{i \in \caI}$ is uniformly rapidly local and localized near $\{ \Sigma_i \}_{i \in \caI}$, then the family $\{ g(U_i) - g(V_i) \}_{i \in \caI}$ is uniformly rapidly local and localized near $\{ \Sigma_i \}_{i \in \caI}$.
\end{proposition}

\begin{proof}
    For the self-adjoint case, let $N \in \N$ and let $\tilde f_{N+1}$ be an almost-analytic extension of $f$ supported in $\Sigma_f := \{ x + \iu y \in \C \, : \, x \in \supp(f),  \abs{y} < 1 \}$. It satisfies $\abs{\partial_{\bar z} \tilde f_{N+1}(z)} \leq c_N \abs{y}^{N+1}$ for some constant $c_N < \infty$. It follows from Proposition \ref{prop:rapidly local Combes-Thomas} that the resolvents $\{ (z - A_i)^{-1} \}_{i \in \caI}$ are uniformly rapidly local with constants $\{ D_N \abs{\Im z}^{-(N + 1)} \}_{N \in \N}$ for all $z \in \C$ with $0 < \abs{\Im z} \leq 1$. Using the Helffer-Sj\"{o}strand formula \eqref{eq:HS for self-adjoints} for self-adjoint operators we therefore get
    \begin{align*}
        \abs{ f(A_i)(x, y) } &\leq \frac{1}{\pi} \int_{\C} \dd x \dd y \, \abs{ \partial_{\bar z} \tilde f_{N+1}(z) } \, \abs{ (z - A_i)^{-1}(x, y) } \\
        &\leq \frac{1}{\pi} \int_{\Sigma_f} \dd x \dd y \, c_N \abs{y}^{N+1} \, D_N \abs{y}^{-(N+1)} \, \jb{x-y}^{-N} \\
        &\leq C'_N \jb{x-y}^{-N}
    \end{align*}
    for some constant $C'_N$ that only depends on $f$ and on $D_N$. This proves the first claim for self-adjoint bounded operators.

    For the second claim, use the Helffer-Sj\"{o}strand formula and the resolvent identity to get
    $$ f(A_i) - f(B_i) = -\frac{1}{\pi} \, \int_{\Sigma_f} \, \dd x \dd y \, \partial_{\bar z} \tilde f_N(z)  \, (z-A_i)^{-1} (A_i - B_i) (z - B_i)^{-1}. $$
    It follows from Proposition \ref{prop:rapidly local Combes-Thomas} and Lemma \ref{lem:norm characterization of rapid locality} that there are constants $\{ D_N \}_{N \in \N}$ so that $\norm{(z - A_i)^{-1}}_N, \norm{(z - B_i)^{-1}}_N \leq D_N \abs{\Im z}^{-(N + d + 2)}$ for all $N \in \N$ and all $z \in \C$ with $0 < \abs{\Im z} \leq 1$. Also by Lemma \ref{lem:norm characterization of rapid locality}, there are constants $\{ C_N \}_{N \in \N}$ so that $\norm{ A_i - B_i }_{N, \Sigma_i} \leq C_N$ for all $N \in \N$. Using items \eqref{eq:sub-multiplicativity} and \eqref{eq:Schur norm property for ideal} from Lemma \ref{lem:*-algebras and ideals of rapidly local operators} in the Helffer-Sj\"{o}strand formula with $N' = 4N + 2d + 4$ we therefore find
    \begin{align*}
        \norm{ f(A_i) - f(B_i) }_{N, \Sigma_i} &\leq \frac{1}{\pi} \int_{\Sigma_f} \, \dd x \dd y \, \abs{ \partial_{\bar z} \tilde f_{N'}(z) } \, \norm{ (z - A_i)^{-1} }_{2N} \,\, \norm{A_i - B_i}_{N, \Sigma_i} \,\, \norm{ (z - B_i)^{-1} }_{2N}  \\
        &\leq\frac{1}{\pi} \int_{\Sigma_f} \, \dd x \dd y \, c_N \abs{y}^{4N + 2d + 4} \, D_{2N}^2 \, \abs{y}^{-(4N + 2d + 4)} \, C_N < \infty
    \end{align*}
    for all $N \in \N$. It now follows from Lemma \ref{lem:norm characterization of rapid locality} that $\{ f(A_i) - f(B_i) \}_{i \in \caI}$ is uniformly local and localized near $\{\Sigma_i \}_{i \in \caI}$, as required.

    The proofs in the unitary case are virtually identical to those of the self-adjoint case, and are left to the reader.
\end{proof}

\subsection{Rapid locality and polar decomposition}

We use the stability of rapid locality and localization under the functional calculus obtained in the previous section to show that this property is also stable under taking the polar part of an invertible operator.

\begin{lemma} \label{lem:A close to U implies pol(A) close to U}
    Let $\caH$ be a separable Hilbert space and $A, U \in \caB(\caH)$ with $U$ unitary such that $\norm{A - U} \leq \delta < 1$. Then $\pol(A)$, the polar part of $A$, is unitary and 
    $$ \norm{ \pol(A) - U } \leq \frac{2\delta}{2 - \delta}. $$
\end{lemma}

\begin{proof}
    This proof is adapted from arguments in \cite{li1995new}. Let $A = V \abs{A}$ be the polar decomposition of $A$. Since $\delta < 1$ we know that $A$ is invertible and therefore $V = \pol(A)$ is unitary.

    Put $E = A - U$ and $W = V^* U$. Then $V^* E = \abs{A} - W$ and $E^*U = \abs{A}W - \I$. It follows that
    $$ V^* E  - E^* U = (\I + \abs{A})(\I - W)$$
    and therefore
    $$ \I - W = (\I + \abs{A})^{-1} ( V^* E - E^* U ). $$
    It follows that
    $$ \norm{V - U} = \norm{ \I - W } \leq \norm{ (\I + \abs{A})^{-1} } \big( \norm{ V^* E } + \norm{E^* U} \big) \leq \frac{2 \delta}{1 + m(A)}$$
    where
    $$ m(A) = \inf_{\norm{\psi} = 1} \norm{  \abs{A} \psi} \geq 1 - \delta. $$
    The claim follows.
\end{proof}

\begin{lemma} \label{lem:rapid locality and polar parts}
    Take $0 \leq \delta < 1$ and let $\{A_i \}_{i \in \caI}$ be a uniformly rapidly local family of operators on $\H{d}{n}$ such that for each $i \in \caI$ there is a unitary $U_i$ such that $\norm{A_i - U_i} \leq \delta$. Then the polar parts $\{ \pol(A_i) \}_{i \in \caI}$ form a family of unitary operators which is uniformly rapidly local and satisfies
    $$ \norm{ \pol(A_i) - U_i} \leq \frac{2\delta}{2 - \delta} \quad \text{for all } \,\, i \in \caI. $$
    If $\{ A_i - \I \}_{i \in \caI}$ is uniformly rapidly local and localized near $\{\Sigma_i\}_{i \in \caI}$, then so is $\{ \pol(A_i) - \I \}_{i \in \caI}$.
\end{lemma}

\begin{proof}
    The claim about the $\norm{\pol(A_i) - U_i}$ is immediate from Lemma \ref{lem:A close to U implies pol(A) close to U}. What we have to show is that $\{ \pol(A_i) \}_{i \in \caI}$ is uniformly rapidly local.

    Since $\norm{A_i - U_i} \leq \delta < 1$ we have
    $$  (1-\delta)^2\I \leq A_i^* A_i \leq (1 + \delta)^2 \I $$
    for all $i \in \caI$. Moreover, it follows from Lemma \ref{lem:*-algebras and ideals of rapidly local operators} that $\{ A_i^* A_i \}_{i \in \caI}$ is uniformly rapidly local. There is a compactly supported smooth function $h : \R \to \R$ such that $h(x) = x^{-1/2}$ for all $x \in [(1-\delta)^2, (1+\delta)^2]$, so it follows from Proposition \ref{prop:rapid locality is preserved by smooth functions} that $\{ \abs{A_i}^{-1} = (A_i^* A_i)^{-1/2} \}_{i \in \caI}$ is uniformly rapidly local. Finally, Lemma \ref{lem:*-algebras and ideals of rapidly local operators} implies that $\{ \pol(A_i) = A_i \abs{A_i}^{-1} \}_{i \in \caI}$ is uniformly rapidly local.

    For the final claim, if $\{ A_i - \I \}$ is uniformly rapidly local and localized near $\{ \Sigma_i \}_{i \in \caI}$ then since $A_i^* A_i - \I = (A_i^* -\I)(A_i - \I) + (A_i^* - \I) + (A_i - \I)$, Lemma \ref{lem:*-algebras and ideals of rapidly local operators} implies that $\{ A_i^* A_i - \I \}_{i \in \caI}$ is also uniformly rapidly local and localized near $\{ \Sigma_i \}_{i \in \caI}$. Now noting that $\I = h(\I)$ it follows from Proposition \ref{prop:rapid locality is preserved by smooth functions} that $\{ \abs{A_i}^{-1} - \I \}_{i \in \caI}$ is uniformly rapidly local and localized near $\{ \Sigma_i \}_{i \in \caI}$. Since $\pol(A_i) - \I = A_i \abs{A_i}^{-1} - \I = (A_i - \I)(\abs{A_i}^{-1} - \I) + (A_i - \I) + (\abs{A_i}^{-1} - \I)$ the final claim follows by Lemma \ref{lem:*-algebras and ideals of rapidly local operators}.
\end{proof}

\subsection{Some applications of the Schur test}

Recall the Schur test
\begin{lemma} \label{lem:Schur test}
    Let $\caH = \bigoplus_{i \in \caI} \caH_i$ and $\caH' = \bigoplus_{i \in \caI} \caH'_i$ be separable Hilbert spaces decomposed as countable direct sums, and let $A \in \caB(\caH, \caH')$ be a bounded operator. Write $A(i, j)$ for the block of $A$ mapping from $\caH_j$ to $\caH'_i$ and suppose
    \begin{align*}
        \sum_j \norm{A(i, j)} &\leq \al \quad \text{for all} \,\, i \in \caI,\\
        \sum_i \norm{A(i, j)} &\leq \beta \quad \text{for all} \, \, j \in \caI.
    \end{align*}
    Then $\norm{A} \leq \sqrt{\al \beta}$.
\end{lemma}

\begin{proof}
    For any vector $\psi \in \caH$ we write $\psi = \bigoplus_{i \in \caI} \psi(i)$ with $\psi(i)$ the component of $\psi$ in $\caH_i$. By the Cauchy-Schwarz inequality we have
    \begin{align*}
        \abs{ (A \psi)(i) }^2 &= \norm{ \sum_j \, A(i, j) \psi(j) }^2 \\
        &\leq \left( \sum_j \norm{A(i, j)} \right)  \left( \sum_j \norm{A(i, j)} \norm{\psi(j)}^2 \right) \\
        &\leq \al \, \sum_{j} \norm{A(i, j)} \norm{\psi(j)}^2
    \end{align*}
    for all $i \in \caI$. Summing this over $i$ yields
    $$ \norm{A \psi}^2 \leq \al \sum_{j} \left( \sum_i \norm{A(i, j)} \right) \norm{\psi(j)}^2 \leq \al \beta \norm{\psi}^2. $$
\end{proof}

The Schur test is a useful tool to derive properties of rapidly local operators. We will use it throughout the manuscript.

Let us here already note the following properties.

\begin{lemma} \label{lem:cutting tails}
    Let $A \in \caB(\H{d}{n})$ be rapidly local with constants $\{ C_N \}_{N \in \N}$. For any $R > 0$ define $A_R$ by the kernel
    $$
        A_R(x, y) = \begin{cases} A(x, y) \quad & \text{if} \, \abs{x-y} \leq R \\
                0 & \text{if}  \, \abs{x-y} > R.
        \end{cases}
    $$
    There are constants $\{D_N\}_{N \in \N}$ which depend only on the $\{ C_N \}$ and on $d$, such that
    $$ \norm{ A - A_R } \leq D_N R^{-N} $$
    for all $N \in \N$.
\end{lemma}

\begin{proof}
    For any $N \in \N$ we have
    $$
        \sum_y \abs{(A -  A_R)(x, y)} = \sum_{y : \abs{x-y} > R} \abs{A(x, y)} \leq (1+R)^{-N} \, \sum_y \, \abs{A(x, y)} \jb{x-y}^N \leq (1+R)^{-N} \norm{A}_{N}.
    $$
    The same bound holds for the sum over $x$ and so the Schur test \ref{lem:Schur test} implies
    $$ \norm{ A - A_R } \leq (1 + R)^{-N} \norm{A}_N \leq C_{N + d + 1} c_d R^{-N} $$
    where we used Lemma \ref{lem:norm characterization of rapid locality}.
\end{proof}

\begin{lemma} \label{lem:distant regions}
    Let $A \in \caB(\H{d}{n})$ be rapidly local with constants $\{ C_N \}_{N \in \N}$. There are constants $\{D_N\}_{N \in \N}$ which depend only on the $\{ C_N \}$ and on $d$, such that
    $$ \norm{ \I_X A \I_Y } \leq D_N \, \big( 1 + \d(X, Y) \big)^{-N} $$
    for all $N \in \N$ and any regions $X, Y \subset \R^d$.
\end{lemma}

\begin{proof}
    Put $R = \d(X, Y)$. For any $N \in \N$ we have
    $$ \sum_y \, \abs{( \I_X A \I_Y)(x, y)} \leq (1 + R)^{-N} \sum_y \abs{A(x, y)} \jb{x-y}^N \leq (1+R)^{-N} \norm{A}_N.$$
    The same bound holds for the sum over $x$ so by the Schur test we get
    $$ \norm{\I_X A \I_Y} \leq \norm{A}_N (1 + R)^{-N} \leq C_{N + d + 1} c_d (1 + R)^{-N}, $$
    where we used Lemma \ref{lem:norm characterization of rapid locality}.
\end{proof}

\begin{lemma} \label{lem:restricting close to localization region}
    Suppose $A \in \caB(\H{d}{n})$ is rapidly local and localized near a region $\Sigma$ with constants $\{ C_N \}_{N \in \N}$. For any $R > 0$ define
    $$  A_R := \I_{\Sigma^{(R)}} A \I_{\Sigma^{(R)}}$$
    where $\Sigma^{(R)} = \{ x \in \R^d \, : \, \d(x, \Sigma) \leq R \}$ is the $R$-fattening of $\Sigma$. Then there are constants $\{ D_N \}_{N \in \N}$ depending only on $\{ C_N \}_{N \in \N}$ and on $d$ such that
    $$ \norm{ A - A_R } \leq D_N R^{-N}. $$
\end{lemma}

\begin{proof}
    For any $N \in \N$ we bound
    $$ 
        \sum_y \abs{(A - A_R)(x, y)} \leq (1 + R)^{-N} \sum_{y} \abs{A(x, y)} \jb{x-y}^N \jb{ \d(\Sigma, x) + \d(\Sigma, y) }^N \leq \norm{A}_{N, \Sigma} (1 + R)^{-N}
    $$
    where we used $\jb{x-y} \geq 1$ and $\jb{ \d(\Sigma, x) + \d(\Sigma, y) } \geq 1 + R$ for all $y$ such that $(A - A_R)(x, y) \neq 0$. The same bound holds for the sum over $x$ so the Schur test \ref{lem:Schur test} gives
    $$ \norm{ A - A_R } \leq \norm{A}_{N, \Sigma} (1 + R)^{-N} \leq C_{N + d +1} c_d R^{-N}, $$
    where we used Lemma \ref{lem:norm characterization of rapid locality}.
\end{proof}

\section{Boundary unitaries and the Chern Number} \label{sec:boundary unitaries and the Chern number}

\subsection{Boundary unitaries in arbitrary dimensions} \label{subsec:boundary unitaries}

Define
\begin{align*}
    \Upper^{(d)}_a &:= \{ x \in \R^d \, : \, x_d \geq a \} \\
    \Lower^{(d)}_a &:= \{ x \in \R^d \, : \, x_d \leq a \} \\
    \Hyper^{(d)}_a &:= \Upper^{(d)}_a \cap \Lower^{(d)}_a = \{ x \in \R^d \, : \, x_d = a\} \\
    \Right_b^{(d)} &:= \{ x \in \R^d \, : \, x_1 \geq b \} \\
    \Left_b^{(d)} &:= \{ x \in \R^d \, : \, x_1 \leq b \} \\
    \sfV_b^{(d)} &:= \Right_b^{(d)} \cap \Left_b^{(d)} = \{ x \in \R^d \, : \, x_1 = b \}.
\end{align*}
When the dimension $d$ is clear from context, we simply write $\Upper_a, \Lower_a$ ,$\Hyper_a$, $\Right_b$, $\Left_b$ and $\sfV_b$.

Let $P$ be a rapidly local projector on $\H{d}{n}$. For any $a \in \R$ let $\Pi_a = \I_{\Upper_a}$ be the spatial projector onto the upper half-space $\Upper_a = \{ x \in \R^d \, : \, x_d \geq a \}$. Put
$$ E_a := P \Pi_a P $$
and define the \emph{boundary unitary} at $a \in \R$ by
$$ U_a(P) = \ed^{2\pi \iu E_a}. $$

In the remainder of this section we exploit the smooth functional calculus machinery introduced in Section \ref{sec:locality and functional calculus} to show that the boundary unitaries are localized near half-planes, and that in the case $d=2$ this allows us to associate an index to the boundary unitaries. Similar results have been obtained with similar techniques in \cite{elbau2002equality}. See also \cite{fonseca2020two, bols2023fredholm}.

\begin{lemma} \label{lem:commutator with half-plane localized near the cut}
    If $\{ A_i \}_{i \in \caI}$ is a uniformly rapidly local family of operators on $\H{d}{n}$ then the family of commutators $\{ [A_i, \Pi_a] \}_{i \in \caI, a \in \R}$ is uniformly rapidly local and localized near $\{ \Hyper_a \}_{a \in \R}$.
\end{lemma}

\begin{proof}
    There are constants $\{ C_N \}_{N \in \N}$ such that $\abs{A_i(x, y)} \leq C_N \jb{x-y}^{-N}$ for all $N$ and all $i \in \caI$. Noting that $( [A_i, \Pi_a] )(x, y) = 0$ unless either $x \in \Upper_a$ and $y \notin \Upper_a$ or vice versa, we get in either case
    $$ \abs{ ([A_i, \Pi_a])(x, y) } = \abs{ A_i(x, y) } \leq C_{2N} \jb{x-y}^{-2N} \leq C_{2N} \, \jb{x-y}^{-N} \jb{ \d(x, \Hyper_a) + \d(y, \Hyper_a) }^{-N}. $$
    We conclude that $\{[A_i, \Pi_a]\}_{i \in \caI, a \in \R}$ is uniformly rapidly local and localized near $\{ \Hyper_a \}_{a \in \R}$.
\end{proof}

\begin{lemma} \label{lem:E is a projector away from the cut}
    Let $P$ be a rapidly local projector on $\H{d}{n}$ and $E_a = P \Pi_a P$ for $a \in \R$. The family $\{ E_a^2 - E_a \}_{a \in \R}$ is uniformly rapidly local and localized near $\{ \Hyper_a \}_{a \in \R}$.
\end{lemma}

\begin{proof}
    Note that $E_a^2 - E_a = P \Pi_a [P, \Pi_a] P$ and that $\{[P, \Pi_a]\}_{a \in \R}$ is uniformly rapidly local and localized near $\{ \Hyper_a \}$ by Lemma \ref{lem:commutator with half-plane localized near the cut}. The claim now follows from Lemma \ref{lem:*-algebras and ideals of rapidly local operators}.
\end{proof}

\begin{proposition} \label{prop:boundary unitaries act near the boundary}
    Let $P$ be a rapidly local projector on $\H{d}{n}$ and $\{ U_a(P) \}_{a \in \R}$ its associated boundary unitaries. Then the family $\{ U_a(P) - \I \}_{a \in \R}$ is uniformly rapidly local and localized near the hyperplanes $\{ \Hyper_a \}_{a \in \R}$.
\end{proposition}

\begin{proof}
    The function $h : [0, 1] \to \C$ defined by
    \begin{align*}
            h(x) = \frac{\ed^{2\pi \iu x} - 1}{x^2 - x}& \quad \text{for} \,\, x \in (0, 1) \\
            h(0) = -2\pi \iu, \quad  & \quad h(1) = 2\pi \iu
    \end{align*}
    is smooth. Let $k : \R \to \C$ be a compactly supported smooth extension of $h$. Then
    $$ U_a(P) - \I = (E_a^2 - E_a) k(E_a)$$
    and since $\{ E_a^2 - E_a \}_{a \in \R}$ is uniformly rapidly local and localized near $\{ \Hyper_a\}_{a \in \R}$ by Lemma \ref{lem:E is a projector away from the cut} we conclude by Proposition \ref{prop:rapid locality is preserved by smooth functions} and Lemma \ref{lem:*-algebras and ideals of rapidly local operators} that the same is true for $\{ U_a(P) - \I \}_{a \in \R}$.
\end{proof}

\subsection{Boundary unitaries in 2D : Flow and bulk-boundary correspondence} \label{subsec:flow and chern number}

Recall the Chern number
$$ \Ch(P) = \ind(P, F_a) \in \Z $$
where $F_a$ is the flux insertion unitary introduced in Section \ref{subsec:Chern number}. It is well-defined whenever $F_a P F_a^* - P$ is compact, which holds for any rapidly local projector $P$ by the argument below \cite[~Lemma 1]{aizenman1998localization}.

To the boundary unitaries $U_a = U_a(P)$ associated to $P$ we can also associate an index called Flow:
\begin{equation} \label{eq:Flow of edge unitaries}
    \Flow(U_a) := \ind( \I_{\Right_b}, U_a ) \in \Z,
\end{equation}
which is independent of $b \in \R$. The required compactness of $U_a \I_{\Right_b} U_a^* - \I_{\Right_b}$ is shown in \cite[Sect. 4.4]{bols2023fredholm}, and Theorem 4.4 of that paper gives the following version of the Bulk-boundary correspondence:

\begin{proposition} [\cite{bols2023fredholm}]\label{prop:bulk-boundary correspondence}
    Let $P$ be a rapidly local projector on $\H{2}{n}$. Then 
    $$ \Ch(P) = \Flow(U_a(P)) $$ for any $a \in \R$.
\end{proposition}

\section{Gapping out rapidly local unitaries in 1D} \label{sec:gapping 1D unitaries}

In this section we work in one dimension, so we consider the Hilbert space $\H{1}{n}$ and for any $a \in \R$ write $\Pi_a = \I_{x \geq a}$ for the projector onto the half-line $\{ x \in \Z \, : \, x \geq a \}$.

To any rapidly local unitary $U \in \caB(\H{1}{n})$ we can associate an integer \emph{flow} given by
$$ \Flow(U) = \ind( \Pi_a, U ) \in \Z $$
which is independent of $a \in \R$. The required compactness of $U \Pi_a U^* - \Pi_a = [U, \Pi_a] U^*$ follows easily from rapid locality of $U$:

\begin{lemma} \label{lem:1D flow well-defined}
    If $U$ is a rapidly local unitary on $\H{1}{n}$ then $U \Pi_a U^* - \Pi_a$ is compact for all $a \in \R$.
\end{lemma}

\begin{proof}
    By Lemma \ref{lem:commutator with half-plane localized near the cut} the commutator $[U, \Pi_a]$ is rapidly local and localized near $\{ a \}$. It therefore follows from Lemma \ref{lem:restricting close to localization region} that $[U, \Pi_a]$ can be approximated in norm by operators with finite support, so $[U, \Pi_a]$ is compact. Since the compact operators form an ideal in $\caB( \H{1}{n} )$, the difference $U \Pi_a U^* - \Pi_a = [U, \Pi_a] U^*$ is also compact.
\end{proof}

In this section we will prove
\begin{proposition} \label{prop:gapping 1D unitaries}
    For every $0 < \ep \leq 1$ and uniformly rapidly local family $\{ U_i \}_{i \in \caI}$ of unitaries on $\H{1}{n}$ with $\Flow(U_i) = 0$ for all $i \in \caI$, there is $R < \infty$ and a family $\{ V_i \}_{i \in \caI}$ of unitaries of propagation $R$ such that
    $$\norm{ U_i - V_i } \leq \ep, \quad \text{and} \quad \dist( -1, \sigma(V_i) ) \geq \ep/2  $$
    for all $i \in \caI$.
\end{proposition}

The proof appears in Section \ref{subsec:proof of gapping 1D unitaries proposition}.

\subsection{Rapidly local decoupling}

The following Lemma shows that rapidly local unitaries can be decoupled across any cut by a perturbation which is rapidly localized near that cut, provided the flow vanishes. Similar decoupling results for 1D unitaries have appeared in \cite[Sect. VII.]{cedzich2018topological} and \cite[~Lemma 3.4.]{chung2025topological}

\begin{lemma} \label{lem:rapidly local decoupling}
    Let $\{U_i\}_{i \in \caI}$ be a uniformly rapidly local family of unitaries on $\H{1}{n}$ with $\Flow(U_i) = 0$ for all $i \in \caI$. For any $a \in \R$ there is a family of unitaries $\{ U_{a, i} \}_{i \in \caI}$ such that $[U_{a, i}, \Pi_a] = 0$ for all $i \in \caI$, and $\{ U_i - U_{a, i} \}_{i \in \caI}$ is uniformly rapidly local and localized near $\{a\}$. 
    
    Moreover, if $[U_{i}, \Pi_{a-R}] = 0$ for all $i \in \caI$ and some $R \geq 1$, then the decouplings at $a$ can be taken so that they preserve the decoupling at $a-R$. That is, $[U_{a, i}, \Pi_{a-R}] = 0$ for all $i \in \caI$.
\end{lemma}

\begin{proof}
    To avoid clutter, we drop $i \in \caI$ from the notation. It will be clear that all estimates are uniform in $i$ and in $a$.
    
    Let $\widetilde \Pi_a = U \Pi_a U^*$ and $A_a = \widetilde \Pi_a - \Pi_a = [U, \Pi_a] U^*$. The commutators $\{ [U, \Pi_a] \}_{a \in \R}$ are uniformly rapidly local and localized near $\{ a \}$ by Lemma \ref{lem:commutator with half-plane localized near the cut}, so it follows from Lemma \ref{lem:*-algebras and ideals of rapidly local operators} that $A_a$ is rapidly local and localized near $\{a\}$ as well.

    We now approximate $\widetilde \Pi_a$ by a projector $\widetilde \Pi_{a, r}$ that acts as 0 far to the left of $a$, and acts as $\I$ far to the right of $a$. To achieve this, let
    $$ B_{a, r} := \Pi_a + \I_{[a-r, a+r]} A_a \I_{[a-r, a+r]} $$
    for any $r \in \N$. Since $A_a$ is rapidly local and localized near $\{a\}$ it follows from Lemma \ref{lem:restricting close to localization region} that $\norm{B_{a, r} - \widetilde \Pi_a} = \norm{ A_a - \I_{[a-r, a+r]} A_a \I_{[a-r, a+r]} } \rightarrow 0$ (uniformly in $i \in \caI, a \in \R$). Take $r$ large enough so that $\norm{ B_{a, r} - \widetilde \Pi_a } \leq 1/4$. Let $f \in C_c^{\infty}$ be such that $f(x) = 0$ for all $x \in [0, 1/4]$ and $f(x) = 1$ for all $x \in [3/4, 1]$. Then $\widetilde \Pi_{a, r} := f( B_{a, r} )$ is a projector with
    $$ \norm{ \widetilde \Pi_{a, r} - \widetilde \Pi_a } \leq \norm{\widetilde \Pi_{a, r} - B_{a, r}} + \norm{ B_{a, r} - \widetilde \Pi_a } \leq 1/2.$$
    The difference $B_{a, r} - \widetilde \Pi_a$ is rapidly local and localized near $\{a\}$.  It therefore follows from Proposition \ref{prop:rapid locality is preserved by smooth functions} that $\widetilde \Pi_{a, r} - \widetilde \Pi_a = f(B_{a, r}) - f(\widetilde \Pi_a)$ is rapidly local and localized near $\{a\}$.

    By the Remark below \cite[~Corollary 2.2]{avron1994index}, if we put
    $$ Z_{a, r} := \widetilde \Pi_{a, r} \widetilde \Pi_{a} + \widetilde \Pi_{a, r}^{\perp} \widetilde \Pi_{a}^{\perp} \quad \text{and} \quad  W_{a, r} := \pol(Z_{a, r}), $$
    then $W_{a, r}$ is unitary and satisfies $ W_{a, r} \widetilde \Pi_{a} = \widetilde \Pi_{a, r} W_{a, r}$. Moreover, we have $Z_{a, r} - \I = (\widetilde \Pi_{a, r} - \widetilde \Pi_{a})( 2 \widetilde \Pi_{a} - \I )$, so that $Z_{a, r} - \I$ is rapidly local and localized near $\{a\}$ by Lemma \ref{lem:*-algebras and ideals of rapidly local operators}. Further noting that
    $$ \norm{ Z_{a, r} - \I } = \norm{ \widetilde \Pi_{a, r}( \widetilde \Pi_{a} - \widetilde \Pi_{a, r} ) + \widetilde \Pi_{a, r}^{\perp}( \widetilde \Pi_{a}^{\perp} - \widetilde \Pi_{a, r}^{\perp} ) } = \norm{ (\widetilde \Pi_{a, r} - \widetilde \Pi_{a, r}^{\perp} )( \widetilde \Pi_{a} - \widetilde \Pi_{a, r} ) } \leq 1/2$$
    we get from Lemma \ref{lem:rapid locality and polar parts} that $W_{a, r} - \I$ is rapidly local and localized near $\{ a \}$.

    Using the intertwining property of $W_{a, r}$ we find that the rapidly local unitary $U'_{a, r} = W_{a, r} U$ satisfies
    $$ U'_{a, r} \Pi_a (U'_{a, r})^* = \widetilde \Pi_{a, r}. $$
    Note that $\widetilde \Pi_{a, r}$ indeed has the property announced above, namely
    $$ \widetilde \Pi_{a, r} = \I_{[a-r, a+r]} \widetilde \Pi_{a, r} \I_{[a-r, a+r]} + \I_{x > a+r}.$$
    Since $U'_{a, r}$ is a compact perturbation of $U$, we have $\Flow(U'_{a, r}) = 0$ and therefore 
    $$ \dim \Ran \I_{[a-r, a+r]} \widetilde \Pi_{a, r} \I_{[a-r, a+r]} = \dim \Ran \I_{[a-r, a+r]} \Pi_{a} \I_{[a-r, a+r]} = \dim \Ran \Pi_{[a, a+r]}. $$
    It follows that there is a unitary $V_{a, r}$ with $V_{a, r}- \I$ supported on $[a-r, a+r]$ such that $V_{a, r} \widetilde \Pi_{a, r} = \Pi_a V_{a, r}$. In particular, $V_{a, r} - \I$ is rapidly local and localized near $\{a\}$. The rapidly local unitary $U_{a, r} = V_{a, r} U'_{a, r} = V_{a,r} W_{a, r} U$ therefore satisfies $[U_{a, r}, \Pi_a] = 0$. Finally, since $U_{a, r} - U = (V_{a, r} W_{a, r} - \I) U = \big(  (V_{a, r} - I)(W_{a, r} - \I) + (V_{a, r} - \I) + (W_{a, r} - \I) \big) U$, it follows from Lemma \ref{lem:*-algebras and ideals of rapidly local operators} that $U_{a, r} - U$ is rapidly local and localized near $\{a\}$.

    Under the additional hypothesis that $[U_i, \Pi_{a-R}] = 0$, it is clear that all the operators $\widetilde \Pi_a, A_a, B_{a, r}, \widetilde \Pi_{a, r}$ are supported on $\Pi_{a-R}$. It follows that $W_{a, r} - \I$ is also supported on $\Pi_{a-R}$ so the unitaries $U'_{a, r}$ still commute with $\Pi_{a-R}$. The final rotation $V_{a, r}$ can then also be done entirely to the right of the cut $a-R$, proving the final claim.
\end{proof}

\subsection{Approximately commuting half-line projectors}

We will construct deformed half-line projectors which approximately commute with the given rapidly local unitary $U$, under the assumption that $\Flow(U) = 0$. The crucial input is the following corollary to Lin's theorem on almost commuting matrices \cite{lin1993almostcommuting, friis1996almost, hastings2009making}:

\begin{lemma}[\cite{herrera2020hastings}~Proposition 16.1] \label{lem:Lin for s.a. and unitary}
    For every $\ep > 0$ there is $\delta > 0$ so that if $a, u \in M_m(\C)$ with $a$ hermitian and $\norm{a} \leq 1$, and $u$ unitary, and $\norm{ [a, u] } \leq \delta$, then there is a unitary $u' \in M_m(\C)$ and a hermitian $a' \in M_m(\C)$ with
    $$ [a', u'] = 0, \quad \norm{a'-a} \leq \ep, \quad \norm{u'-u} \leq \ep.$$
\end{lemma}
Note that the threshold $\delta$ depends only on $\ep$, not on the dimension $m$.

We will use this Lemma to prove the following result, whose proof resembles the proofs of \cite[~Lemma 3]{hastings2009making} and \cite[~Lemma 11.1]{herrera2020hastings}. See also \cite{davidson1985almost}.

\begin{lemma} \label{lem:approximately commuting half-line projectors}
    Let $U$ be a unitary on $\H{1}{n}$ which is rapidly local with constants $\{ C_N \}_{N \in \N}$. Suppose further that $\Flow(U) = 0$.
    Then for every $\ep > 0$ there exists $R < \infty$ depending only on $\ep$ and on $\{ C_N \}_{N \in \N}$ so that for every $a \in \R$ there is an orthogonal projection $E_a$ which satisfies
    $$ E_a \I_{\{x < a-R\}} = 0, \quad E_a \I_{\{x > a+R\}} = \I_{\{x > a+R\}}, \quad \norm{ [U, E_a] } \leq \ep. $$
\end{lemma}

\begin{proof}
    The first step is to use Lemma \ref{lem:rapidly local decoupling} to decouple $U$ at the endpoints of a sufficiently large interval $[a-R, a+R]$.

    A first application of Lemma \ref{lem:rapidly local decoupling} yields a unitary $U_{a-R}$ such that $[U_{a-R}, \Pi_{a-R}] = 0$ and $U_{a-R} - U$ is rapidly local and localized near $\{ a - R \}$, uniformly in $a$ and $R$. Applying Lemma \ref{lem:rapidly local decoupling} to $U_{a-R}$ yields a unitary $\widetilde U$ such that
    $$ [\widetilde U, \Pi_{a - R}] = [\widetilde U, \Pi_{a+R}] = 0 $$
    and such that $\widetilde U - U_{a-R}$ is rapidly local and localized near $\{a+R\}$. In particular,
    \begin{equation} \label{eq:U tilde - U is localized near the cuts}
        \widetilde U - U = (\widetilde U - U_{a-R}) + (U_{a-R} - U) =: K_{a+R} + K_{a-R}
    \end{equation}
    is a sum of two terms, the first of which is rapidly local and localized near $\{a+R\}$, and the second of which is rapidly local and localized near $\{a-R\}$, all uniformly in $a$ and $R$.

    Write $I = [a-R, a+R)$, $I_L = \{x < a-R \}$, and $I_R = \{x \geq a+R\}$. The unitary $\widetilde U$ is decoupled at $a-R$ and at $a+R$. That is,
    $$ \widetilde U = \begin{bmatrix}
        \widetilde U_{I_L} && && \\   &&  \widetilde U_{I} &&  \\ 
          &&   && \widetilde U_{I_R}
          \end{bmatrix}$$
    according to the direct sum decomposition
    $$ \H{1}{n} = \I_{I_L} \H{1}{n} \oplus \I_{I} \H{1}{n} \oplus \I_{I_R} \H{1}{n}.$$

    Consider the function $\lambda : I \to [0, 1]$ given by (see Figure \ref{fig:ramp function})
    $$ \lambda(x) = \begin{cases} 
                    0 \quad\text{if} \, x \in [a-R, a-R/2) \\
                    \frac{x - a + R/2}{R} \quad \text{if} \, x \in [a-R/2, a+R/2] \\
                    1 \quad \text{if} \, x \in (a+R/2, a+R)
    \end{cases} $$
    and define the multiplication operator $\Lambda \delta_{x, i} = \lambda(x) \delta_{x, i}$ on $\I_{I} \H{1}{n}$.

    \begin{figure}
    \begin{center}
        \includegraphics[width=14cm]{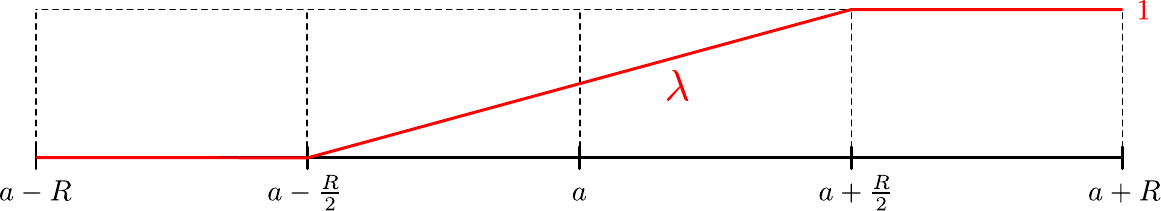}
    \end{center}
    \caption{The function $\lambda : I \to [0, 1]$ is identically zero on the left buffer region $[a-R, a-R/2)$ and identically one on the right buffer region $(a+R/2, a+R)$. In the middle region, it rises linearly from 0 to 1, with a slope $R^{-1}$ which is small for large values of $R$.}
    \label{fig:ramp function}
\end{figure}

    Since $\widetilde U$ is rapidly local, the Schur test \ref{lem:Schur test} implies that there is a constant $C < \infty$ independent of $a$ and $R$ such that
    $$ \norm{ [ \widetilde U_{I}, \Lambda ] } \leq \frac{C}{R}. $$

    For a given $\ep_0 > 0$, we can therefore take $R$ large enough so that Lemma \ref{lem:Lin for s.a. and unitary} gives a unitary $U'$ and a self-adjoint $\Lambda'$ on $\I_{[a-R, a+R)} \H{1}{n}$ such that
    $$ [U', \Lambda'] = 0, \quad \norm{\Lambda - \Lambda'}, \norm{U' - \widetilde U_{I}} \leq \ep_0. $$

    We now turn $\Lambda'$ into a projector while retaining the approximate commutativity with $\widetilde U$, and approximate equality with $0$ and $\I$ on the collar regions $[a-R, a-R/2)$ and $(a+R/2, a+R)$ respectively.

    Define the projector $F := \chi_{[1/2, \infty)}(\Lambda')$. Then $[U', F] = 0$ and
    \begin{equation} \label{eq:approximate commutation of F and U tilde}
        \norm{[F, \widetilde U_I]} = \norm{ [F, \widetilde U_I - U'] + [F, U'] } \leq 2 \ep_0.
    \end{equation}
    Let's give names $L = [a-R, a-R/2)$ and $R = (a+R/2, a+R)$ to the left and right collars of $I$, and let $M = [a-R/2, a+R/2]$ be the middle region. On $\Ran F$ we have $\Lambda' \geq 1/2$ by definition, so $\Lambda'$ is invertible on $\Ran F$ with inverse bounded by 2. This observation together with
    $$ F \I_L = (\Lambda'|_{\Ran F})^{-1} F \Lambda' \I_L = (\Lambda'|_{\Ran F})^{-1} F (\Lambda' - \Lambda) \I_L, $$
    implies that
    \begin{equation} \label{eq:F is close to 0 on L}
        \norm{ F \I_L } \leq 2 \ep_0.
    \end{equation}
    A similar argument shows that
    \begin{equation} \label{eq:F is close to 1 on R}
        \norm{ F^{\perp} \I_R } \leq 2\ep_0.
    \end{equation}
    These bounds express that $F$ approximately acts as 0 and $\I$ on $L$ and $R$ respectively. It follows that we can approximate $F$ by $\I_{M} F \I_M + \I_R$. Indeed,
    \begin{align*}
        F &= \I_L F \I_L + \I_L F \I_M + \I_L F \I_R \\
        &+ \I_M F \I_L + \I_M F \I_M + \I_M F \I_R \\
        &+ \I_R F \I_L + \I_R F \I_M + \I_R F \I_R
    \end{align*}
    where any term containing a factor $F \I_L$ is bounded by $2 \ep_0$, and 
    $$ \norm{ \I_R F \I_M  } = \norm{ \I_R \I_M - \I_R F^{\perp} \I_M} \leq 2\ep_0. $$
    Together with
    $$ \norm{  \I_R F \I_R - \I_R} = \norm{ \I_R F^{\perp} \I_R } \leq 2\ep_0 $$
    we conclude from this that
    $$ \norm{ F - (\I_M F \I_M + \I_R) } \leq 16 \ep_0.$$

    Now define a new projector $G = \chi_{[1/2, \infty)}(\I_M F \I_M  + \I_R)$. It satisfies
    $$ G \I_L = 0, \quad G \I_R = \I_R, \quad \norm{ [G, \widetilde U_I] } \leq  66 \ep_0$$
    where we used Eq. \eqref{eq:approximate commutation of F and U tilde}. Let
    $$
    E_a :=  \begin{bmatrix}
         0 && && \\
         &&  G &&  \\ 
          &&   && \I
          \end{bmatrix}.
    $$
    Then $\norm{ [E_a, \widetilde U] } \leq 66 \ep_0$. Let
    \begin{align*}
        K_{a+R}^{(R)} &:=  \I_{(a+R/2, a+3R/2)} K_{a+R} \I_{(a+R/2, a+3R/2)} \\
        K_{a-R}^{(R)} &:= \I_{(a-3R/2, a-R/2)} K_{a-R} \I_{(a-3R/2, a-R/2)}
    \end{align*}
    so $[K_{a+R}^{(R)}, E_a] = [K_{a-R}^{(R)}, E_a] = 0$. From Eq. \eqref{eq:U tilde - U is localized near the cuts} and Lemma \ref{lem:restricting close to localization region}
    we have
    $$ \norm{ (\widetilde U - U) - ( K_{a+R}^{(R)} +  K_{a-R}^{(R)})  } \leq D_N R^{-N} $$
    for all $N \in \N$ and constants $\{ D_N \}_{N \in \N}$ which are independent of $a$ and $R$. It follows that
    $$ \norm{[ \widetilde U - U, E_a ]} \leq 2D_N R^{-N} $$
    and hence
    $$ \norm{ [U, E_a] } \leq 66 \ep_0 + 2D_N R^{-N} $$
    for all $N \in \N$. Choosing $\ep_0 = \ep/132$ and taking $R$ large enough yields the claim.
\end{proof}

\subsection{Proof of Proposition \ref{prop:gapping 1D unitaries}} \label{subsec:proof of gapping 1D unitaries proposition}
We use the almost-commuting half-line projections provided by Lemma \ref{lem:approximately commuting half-line projectors} to show that each $U_i$ is norm close to a direct sum of localized finite dimensional unitaries. These can then be deformed to open up a gap around -1. The first part of the construction resembles the proof of \cite[~Lemma 2]{hastings2009making}.

By hypothesis each $U_i$ is rapidly local with constants $\{ C_N \}_{N \in \N}$ independent of $i$. To avoid clutter we drop $i$ from the notation and make sure all bounds depend only on the constants $\{ C_N \}$.

For any $\ep_0 > 0$, Lemma \ref{lem:approximately commuting half-line projectors} provides a uniform $R < \infty$ and projectors $\{ E_a \}_{a \in \R}$ such that
$$ E_a \I_{\{x < a-R\}} = 0, \quad E_a \I_{\{x > a+R\}} = \I_{\{x > a+R\}}, \quad \norm{ [U, E_a] } \leq \ep_0 $$
for all $a \in \R$. Fix $L > 4R$ and let $F_j := E_{j L}$ for all $j \in \Z$. Then
$$ F_j F_{j+1} = F_{j+1} = F_{j+1} F_j, \quad \text{for all} \, j \in \Z. $$
It follows that $\{ D_j := F_{j} - F_{j+1} \}_{j \in \Z}$ is a family of mutually orthogonal projectors and that $\sum_j D_j = \I$. Moreover, each $D_j$ is supported on $[jL - R, (j+1)L + R]$, an interval of length $L + 2R$.

Define
$$ B := \sum_j D_j U D_j $$
and note that $D_j U D_{j+1} = D_j [U, F_{j+1}] D_{j+1}$ from which it follows that
$$ \norm{D_j U D_{j+1}} \leq \ep_0 \quad \text{for all} \, j \in \Z$$
and by the same reasoning
$$ \norm{D_{j+1} U D_{j}} \leq \ep_0 \quad \text{for all} \, j \in \Z.$$

If $\abs{j - k} \geq 2$ then the supports of $D_j$ and $D_k$ are separated by a distance at least $\delta(j, k) := (\abs{j - k} - 1)L - 2R \geq \abs{j - k} L /4$. It therefore follows from Lemma \ref{lem:distant regions} that there are constants $\{D'_N\}_{N \in \N}$ depending only on $\{C_N\}_{N \in \N}$ such that
$$ \norm{ D_j U D_k } \leq D'_N (\abs{j-k} L)^{-N} $$
whenever $\abs{j-k} \geq 2$.

Write $T_{jk} = D_j (U - B) D_k$ for the matrix blocks of $U - B$ with respect to the decomposition $\H{1}{n} = \bigoplus_j D_j \H{1}{n}$. Then the above bounds yield
$$ \sum_k \norm{ T_{jk} } \leq 2\ep_0 + 2 D'_N L^{-N} \sum_{q \geq 2} q^{-N} $$
where the sum on the right hand side is finite for $N \geq 2$. A similar bound holds for the sum over $j$ so the Schur test yields
$$\norm{U - B} \leq  2\ep_0 + D_{N} L^{-N}$$
for all $N \in \N$ and constants $\{ D_N \}_{N \in \N}$ depending only on $\{ C_N \}_{N \in \N}$.

Choose $\ep_0 = \ep/16$ and take $L$ large enough so that $\norm{U - B} \leq \ep / 4$. Since $B$ is reduced by all of the $D_j$, so is its polar part $W = \pol(B)$, and by Lemma \ref{lem:A close to U implies pol(A) close to U} we have
$$ \norm{U - W} \leq \ep/2. $$

Finally, perturb each block $D_j W D_j$ to obtain a unitary $V$ so that $\dist(-1, \sigma(V)) \geq \ep/2$. This can be done so that $\norm{W - V} \leq \ep/2$ and therefore $\norm{U - V} \leq \ep$. Finally, note that $V$ has propagation bounded by $L + 2R + 2$, the upper bound on the size of the supports of the $D_j$. \hfill$\qedsymbol$\vskip.5cm

\begin{remark}\label{rem:real rank zero}
    The original version of Lin's theorem is for pairs of almost commuting hermitian matrices \cite{lin1993almostcommuting}. Using this result it is straightforward to repeat the analysis of this section for rapidly local self-adjoint operators on $\H{1}{n}$. The decoupling step in Lemma \ref{lem:rapidly local decoupling} is trivial in the self-adjoint case because diagonal blocks of self-adjoints are still self-adjoint. In particular no topological flow obstruction appears.

    The self-adjoint analogue of the proof of Proposition \ref{prop:gapping 1D unitaries} then establishes that any rapidly local self-adjoint operator $H$ on $\H{1}{n}$ can be arbitrarily well approximated in norm by a direct sum of finite-dimensional blocks. These blocks can then be deformed so that their spectra all lie in the same finite set. In this way we see that any such $H$ is norm-close to an operator with finite spectrum. This gives an artisanal proof of the fact that the uniform Roe algebra over $\Z$ has real rank zero, \ie the set of its self-adjoint elements with finite spectrum is norm-dense in the set of all self-adjoint elements. This result was recently proven by general $\rm C^*$-algebraic means in \cite[~Corollary 4.7]{an2026complexity}. Interestingly, they also use Lin's theorem.
\end{remark}

\section{Rapidly localized Wannier bases in 2D} \label{sec:rapidly localized Wannier bases in 2D}

Our ultimate goal is to find exponentially localized Wannier bases for exponentially local projectors on $\H{2}{n}$ with zero Chern number. The proof strategy is to first construct \emph{rapidly localized} Wannier bases (Section \ref{subsec:2D rapidly local Wannier}), and then upgrade them to exponential locality in the final step (Section \ref{sec:rapid to exponential}). Let us here define what we mean by a rapidly localized Wannier basis.

\begin{definition} \label{def:rapidly localized Wannier basis}
    Let $P \in \caB(\H{d}{n})$ be a projector. A rapidly localized Wannier basis for $P$ with constants $\{ C_N \}_{N \in \N}$ is an orthonormal basis $\{ w_{\al} \}_{\al \in \caA}$ of $\Ran P$ for which there are localization centres $c_{\al} \in \R^d, \al \in \caA$ such that
    $$ \abs{w_{\al}(x)} \leq C_N \jb{x-c_{\al}}^{-N} \quad \text{for all } \, x \in \Z^d, \al \in \caA.  $$
\end{definition}

In this section we will show that any Chern trivial rapidly local projector $P$ on $\H{2}{n}$ admits a rapidly local Wannier basis.

The idea is to decompose $P$ as a sum $P = \sum_{j \in \Z} P_j$ of mutually orthogonal and rapidly local `strip projections' $P_j$, each of which is localized near a horizontal strip. Such rapidly local one dimensional projections will be shown to admit rapidly localized Wannier bases in Section \ref{subsec:1D rapidly local Wannier}. These Wannier bases can be combined into a rapidly local Wannier basis for $P$.

\subsection{Gapping out the boundary unitaries}

Let $P$ be a rapidly local projection on $\H{2}{n}$ with $\Ch(P) = 0$.

We apply Proposition \ref{prop:gapping 1D unitaries} to gap out the boundary unitaries $U_a = \ed^{2 \pi \iu E_a}$ with $E_a = P \Pi_a P$ introduced in Section \ref{subsec:boundary unitaries}.

\begin{lemma} \label{lem:gapping out boundary unitaries}
    Let $\{U_a\}_{a \in \R}$ be the family of boundary unitaries associated to $P$. For every $0 < \ep \leq 1$ there are integers $L, R < \infty$ and a family of unitaries $\{ V_a \}_{a \in \R}$ such that each $V_a$ has propagation bounded by $R$, has $V_a - \I$ supported on the strip $\sfS_a = \Upper_{a-L} \setminus \Upper_{a+L}$, and satisfies
    $$ \norm{ U_a - V_a } \leq \ep, \quad \dist(-1, \sigma(V_a)) \geq \ep/4. $$
\end{lemma}

\begin{proof}
    Put
    $$ B_a := \I_{\sfS_a} U_a \I_{\sfS_a} + \I_{\sfS_a}^{\perp} = \I + \I_{\sfS_a}( U_a - \I) \I_{\sfS_a}. $$
    Since the $\{ U_a - \I \}$ are uniformly rapidly local and rapidly localized near the lines $\Hyper_a$ by Proposition \ref{prop:boundary unitaries act near the boundary}, it follows from Lemma \ref{lem:restricting close to localization region} that there are constants $\{D_N\}_{N \in \N}$ such that
    $$ \norm{ B_a - U_a } \leq D_N L^{-N} $$
    for all $N \in \N$. By Lemma \ref{lem:rapid locality and polar parts} we can therefore take $L$ large enough so that $\{W_a = \pol(B_a)\}_{a \in \R}$ is uniformly rapidly local and
    $$ \norm{ W_a - U_a } \leq \ep/6. $$
    Since $B_a - \I$ is supported on $\sfS_a$ the same is true for $W_a - \I$. In other words, $W_a = \I_{\sfS_a} W_a \I_{\sfS_a} + \I^{\perp}_{\sfS_a}$, where the $\I_{\sfS_a} W_a \I_{\sfS_a}$ form a uniformly rapidly local family of unitaries on the one-dimensional system $\I_{\sfS_a} \H{2}{n} \simeq \H{1}{2Ln}$. Since $\Flow(U_a) = 0$, also $\Flow(W_a) = 0$ for all $a \in \R$. Indeed, \cite[~Proposition 5.1.]{avron1994index} gives
    \begin{align*}
        \Flow(W_a) &= \ind(\I_{\Right_b}, W_a) = \ind(\I_{\Right_b}, U_a) + \ind( U_a \I_{\Right_b} U_a^*, W_a U_a^* ) \\
        &= \Flow(U_a) + \ind( U_a \I_{\Right_b} U_a^*, W_a U_a^* )
    \end{align*}
    and $\ind( U_a \I_{\Right_b} U_a^*, W_a U_a^* ) = 0$ because $\norm{W_a \I_{\Right_b} W_a^* - U_a \I_{\Right_b} U_a^* } \leq \ep/3 < 1$.
    
    We can therefore apply Proposition \ref{prop:gapping 1D unitaries} with $\ep/2$ to obtain a propagation bound $R'$ and a family of unitaries $\{V_a\}$ of propagation $R'$ (in the Euclidean metric of $\H{1}{2Ln}$) so that $V_a - \I$ is supported on $\sfS_a$ and
    $$ \norm{W_a - V_a} \leq \ep/2, \quad \dist(-1, \sigma(V_a)) \geq \ep/4. $$
    The propagation bound $R'$ w.r.t. the Euclidean metric for $\H{1}{2Ln}$ translates to a propagation bound $R = \ceil{\sqrt{ (R')^2 + (2L)^2 }}$ for the $V_a$ as unitaries on $\H{2}{n}$.
    
    Combining this with  $\norm{U_a-W_a} \leq \ep/2$ we get $\norm{U_a - V_a} \leq \ep$, concluding the proof.
\end{proof}

\subsection{Spectral gluing}

We would like to lift the gapped approximations $V_a$ of the boundary unitaries $U_a$ obtained in Lemma \ref{lem:gapping out boundary unitaries} to obtain uniformly gapped approximations $G_a$ to the operators $E_a = P \Pi_a P$. If this can be done, then we can use the smooth functional calculus to uniformly turn the $G_a$ into approximate half-space projectors. In order for this to work, we must first modify the $V_a$ so that they are equal to $U_a$ on a spectral window around 1, while maintaining the gap around -1.

Fix a smooth function $s : \T \to [0, 1]$ such that
\begin{equation}
    s(z) = 1 \,\, \text{if} \,\, \abs{z+1} \leq 1/2 \quad\quad \text{and} \quad\quad s(z) = 0 \,\, \text{if} \,\, \abs{z-1} \leq 1/2,
\end{equation}
and $0 < s(z) < 1$ everywhere else. That is, $s$ is a smooth filter function on a spectral window around -1, see Figure \ref{fig:filter function}.

\begin{figure}
    \begin{center}
        \includegraphics[width=6cm]{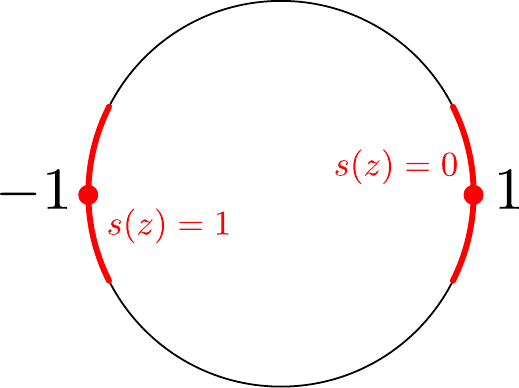}
    \end{center}
    \caption{The smooth filter function $s : \T \to [0, 1]$ takes the value 1 in a neighbourhood of -1, and the value 0 in a neighbourhood of 1.}
    \label{fig:filter function}
\end{figure}

\begin{lemma} \label{lem:spectral gluing}
    Let $\{ U_{i} \}_{i \in \caI}$ and $\{ V_{i} \}_{i \in \caI}$ be uniformly rapidly local families of unitaries on $\H{d}{n}$ such that $\{ U_i - \I \}_{i \in \caI}$ and $\{ V_i - \I \}_{i \in \caI}$ are uniformly rapidly local and localized near some $\{ \Sigma_i \}_{i \in \caI}$.

    Suppose there is $0 < \delta < 1/736$ such that
    $$ \norm{U_i - V_i} \leq \delta, \quad \dist(-1, \sigma{V_i}) \geq \delta/4 $$
    for all $i \in \caI$. Then the operators
    \begin{equation}
        K_i = K_i^* = -\iu \log( V_i U_i^* ), \quad F_i = s(U_i), \quad W_i = \ed^{\iu F_i K_i F_i} U_i
    \end{equation}
    are well-defined for all $i \in \caI$, where $\log$ denotes the principal logarithm, and
    \begin{equation} \label{eq:spectral gluing closeness results}
        \norm{U_i - W_i} \leq 2\delta, \quad \dist( -1, \sigma(W_i) ) \geq \delta/8.
    \end{equation}
    The family $\{ W_i - \I \}_{i \in \caI}$ is uniformly rapidly local and localized near $\{ \Sigma_i \}_{i \in \caI}$.
    
    Furthermore, if $Q_i$ is the projector onto $\Ker F_i$ then $U_i Q_i = W_iQ_i$ and $[U_i, Q_i] = [W_i, Q_i] = 0$, and if there is a subspace $\caK \subset \H{d}{n}$ such that $U_i|_{\caK} = \I_{\caK}$, then also $W_i|_{\caK} = \I_{\caK}$.
\end{lemma}

\begin{proof}
    Let's drop $i \in \caI$ from the notation, and note that all the bounds in the proof are uniform.

    Since $\norm{VU^* - \I} = \norm{V - U} \leq \delta < 1$ the unitary $VU^*$ has its spectrum contained in the ball $\{ z \in \C \, : \, \abs{z-1} \leq \delta \}$. Since $\delta < 2$, the logarithm is well-defined and
    $$\norm{K} \leq 2 \arcsin(\delta/2) \leq \frac{\pi}{2} \delta \leq 2\delta.$$

    The operator $F = s(U)$ is self-adjoint and $0 \leq F \leq \I$ so
    $$ \norm{W - U} = \norm{ \ed^{\iu F K F} - \I } \leq \norm{FKF} \leq 2 \delta.$$

    Since $\Ker F$ reduces $U$ and $F$ vanishes on $\Ker F$, we find that $W|_{\Ker F} = U|_{\ker F}$. Moreover, if U acts as identity on $\caK \subset \H{d}{n}$ then since $s(1) = 0$ we have $\caK \subset \Ker F$ and therefore $W|_{\caK} = \I_{\caK}$.

    Let us now prove the gap assertion in \eqref{eq:spectral gluing closeness results}. Put
    $$ E = \chi_{ \{ \abs{z+1} \leq 1/2 \} }(U). $$
    It is the spectral projector to which $F$ is a smooth approximation. Since $s \chi_{\{ \abs{z+1} \leq 1/2 \}} = \chi_{\{ \abs{z+1} \leq 1/2 \}}$ we have $F E = E$ so
    $$ E( FKF - K )E = 0. $$
    Using $V = \ed^{\iu K} U$ and $W = \ed^{\iu FKF} U$ and expanding the exponentials
    $$ \norm{ E( W - V ) E } \leq 2 \sum_{n \geq 2} \frac{(2\delta)^n}{n!} \leq 8 \left( \sum_{n \geq 2} \frac{1}{n!} \right)  \, \delta^2 \leq  6 \delta^2 $$
    where we used $\delta \leq 1/2$.

    Because $[E, U] = 0$, it follows from $\norm{U - V} \leq \delta$ and $\norm{U - W} \leq 2\delta$ that
    $$ \norm{ E V E^{\perp} }, \norm{E^{\perp} V E} \leq \delta \quad \text{and} \quad  \norm{ E W E^{\perp} }, \norm{E^{\perp W E}} \leq 2\delta.$$

    To show that $W$ inherits the gap of $V$ we want to show that $W - z$ is invertible for $z$ sufficiently close to -1. We will do this by showing that the Schur complement of $W - z$ w.r.t. the decomposition $\H{d}{n} = E \H{d}{n} \oplus E^{\perp} \H{d}{n}$ is invertible. To this end, write
    $$ U - z = \begin{bmatrix} A_U(z) && 0 \\ 0 && D_U(z) \end{bmatrix}, \quad V - z = \begin{bmatrix} A_V(z) && B_V \\ C_V && D_V(z) \end{bmatrix}, \quad W - z = \begin{bmatrix} A_W(z) && B_W \\ C_W && D_W(z) \end{bmatrix}. $$
    Then the above estimates give us
    $$ \norm{B_V}, \norm{C_V} \leq \delta, \quad \norm{B_W}, \norm{C_W} \leq 2\delta $$
    and
    $$ \norm{ A_W(z) - A_V(z) } \leq 6 \delta^2 $$
    for all $z \in \C$.

    The spectrum of $U|_{\Ran E^{\perp}}$ is contained in $\{ \abs{w+1} \geq 1/2 \}$. For $\abs{z+1} \leq 1/4$ we therefore have $\dist \big(z, \sigma( U|_{\Ran E^{\perp}}) \big) \geq 1/4$ and so $D_U(z) = U|_{\Ran E^{\perp}} - z$ is invertible with $\norm{ D_U(z)^{-1} } \leq 4$. This implies bounds on $D_V(z)^{-1}$ and $D_W(z)^{-1}$ as follows. On $E^{\perp} \H{d}{n}$ define
    $$ H_V(z) = D_U(z)^{-1}( D_V(z) - D_U(z) ) \quad \text{and} \quad H_W(z) = D_U(z)^{-1}( D_W(z) - D_U(z) )$$
    Then $D_V(z) = D_U(z)(\I + H_V(z))$ and $D_W(z) = D_U(z)(\I + H_W(z))$ and we get
    $$\norm{ H_V(z) } \leq 4\delta \quad \text{and} \quad \norm{ H_W(z) } \leq 8 \delta.$$
    Since $\delta \leq 1/16$ the geometric series for $(\I + H_V(z))^{-1}$ and $(\I + H_W(z))^{-1}$ converge, yielding bounds
    \begin{align*}
        \norm{D_V(z)^{-1}} &\leq \norm{ D_U(z)^{-1} } \norm{(\I + H_V(z))^{-1}} \leq \frac{4}{1 - 4\delta} \leq 8 \\
        \norm{D_W(z)^{-1}} &\leq \norm{ D_U(z)^{-1} } \norm{(\I + H_W(z))^{-1}} \leq \frac{4}{1 - 8\delta} \leq 8
    \end{align*}
    uniformly for $\abs{z+1} \leq 1/4$.

    The Schur complements
    \begin{align*}
        S_V(z) &= A_V(z) - B_V D_V(z)^{-1} C_V \\
        S_W(z) &= A_W(z) - B_W D_W(z)^{-1} C_W
    \end{align*}
    are therefore well-defined and
    \begin{align*} 
        \norm{ S_W(z) - S_V(z) } &\leq \norm{ (A_W(z) - A_V(z)) } + \norm{ B_W D_W(z)^{-1} C_W } + \norm{ B_V D_V(z)^{-1} C_V } \\
        &\leq 6\delta^2 + 32 \delta^2 + 8 \delta^2 = 46 \delta^2.
    \end{align*}
    for $\abs{z+1} \leq 1/4$.

    Now note that
    $$ \begin{bmatrix} A_V(z) && B_V \\ C_V && D_V(z)  \end{bmatrix} = \begin{bmatrix} \I && B_V D_V(z)^{-1} \\ 0 && \I  \end{bmatrix} \begin{bmatrix} S_V(z) && 0 \\ 0 && D_V(z)  \end{bmatrix} \begin{bmatrix} \I && 0 \\ D_V(z)^{-1} C_V && \I  \end{bmatrix}$$
    where the block triangular operators are invertible (flip the sign of the off-diagonal blocks). For $\abs{z+1} < \delta/8$ we have that $V - z$ is invertible with $\norm{(V - z)^{-1}} \leq 8/\delta$ due to the gap assumption on $V$. It follows that for such $z$ also $S_V(z)$ is invertible, and from the above we have
    $$ (V - z)^{-1} = \begin{bmatrix} S_V(z)^{-1} && - S_V(z)^{-1} B_V D_V(z)^{-1}  \\  - D_V(z)^{-1} C_V S_V(z)^{-1} && D_V(z)^{-1} + D_V(z)^{-1} C_V S_V(z)^{-1} B_V D_V(z)^{-1} \end{bmatrix} $$
    from which it follows that
    $$\norm{ S_V(z)^{-1} } \leq 8/\delta $$
    for $\abs{z+1} < \delta/8$ as well.

    Now
    $$ S_W(z) = S_V(z) \big( \I + S_V(z)^{-1}( S_W(z) - S_V(z) ) \big)$$
    and
    $$ \norm{ S_V(z)^{-1}( S_W(z) - S_V(z)) } \leq 8 \times 46 \delta = 368 \delta.$$
    It follows that with $\delta \leq 1/736$ the Neumann series for $S_W(z)^{-1}$ converges whenever $\abs{z+1} < \delta/8$. That is, the Schur complement $S_W(z)$ is invertible and therefore $W - z$ has inverse
    $$ (W - z)^{-1} = \begin{bmatrix} S_W(z)^{-1} && - S_W(z)^{-1} B_W D_W(z)^{-1}  \\  - D_W(z)^{-1} C_W S_W(z)^{-1} && D_W(z)^{-1} + D_W(z)^{-1} C_W S_W(z)^{-1} B_W D_W(z)^{-1} \end{bmatrix}.$$
    We conclude that $\abs{z+1} < \delta/8$ implies that $z \notin \sigma(W)$, as was to be shown.

    Let us finally show that the family $\{W_i - \I\}_{i \in \caI}$ is uniformly rapidly local and localized near $\{ \Sigma_i \}_{i \in \caI}$. Let $T_i = V_i U_i^*$ and recall that $\sigma(T_i)$ is contained in $\{ \abs{z-1} \leq 1/2 \}$. We have
    $$ T_i - \I = \big( (V_i - \I) - (U_i - \I)  \big) U_i^* $$
    and both $\{ V_i - \I \}_{i \in \caI}$ and $\{ U_i - \I \}_{i \in \caI}$ are uniformly rapidly local and localized near $\{\Sigma_i \}_{i \in \caI}$ while $\{ U_i^* \}_{i \in \caI}$ is uniformly rapidly local. It therefore follows from Lemma \ref{lem:*-algebras and ideals of rapidly local operators} that $\{ T_i - \I \}_{i \in \caI}$ is uniformly rapidly local and localized near $\{ \Sigma_i \}_{i \in \caI}$.

    We now want to show that the $K_i = -\iu \log T_i$ are uniformly rapidly local and localized near $\{ \Sigma_i \}$. To this end let $h : \T \to \C$ be a smooth function which agrees with $\log(z)$ on $\{ \abs{z-1} < 1/2\}$. Since $h(1) = 0$ we have
    $$ K_i = -\iu ( h(T_i) - h(\I)). $$
    Therefore $\{ K_i \}_{i \in \caI}$ is uniformly rapidly local and localized near $\{ \Sigma_i \}_{i \in \caI}$ by Proposition \ref{prop:rapid locality is preserved by smooth functions}.

    Proposition \ref{prop:rapid locality is preserved by smooth functions} also implies that $\{ F_i = s(U_i) \}_{i \in \caI}$ is uniformly rapidly local and therefore $\{ F_i K_i F_i \}_{i \in \caI}$ is uniformly rapidly local and localised near $\{ \Sigma_i \}_{i \in \caI}$ by another application of Lemma \ref{lem:*-algebras and ideals of rapidly local operators}.

    We now show that $\{ \ed^{\iu F_iK_iF_i} - \I \}_{i \in \caI}$ is uniformly rapidly local and localized near $\{ \Sigma_i \}_{i \in \caI}$. To this end, let $k : \R \to \C$ be a compactly supported smooth function which agrees with $\ed^{\iu z}$ on $[ -2\delta, 2\delta ] \supset \sigma(F_i K_i F_i)$. Since $k(0) = 1$ we have
    $$ \ed^{\iu F_i K_i F_i} - \I = k( F_i K_i F_i ) - k(0)$$
    so $\{ \ed^{\iu F_iK_iF_i} - \I \}_{i \in \caI}$ is uniformly rapidly local and localized near $\{ \Sigma_i \}_{i \in \caI}$ by Proposition \ref{prop:rapid locality is preserved by smooth functions}. Finally,
    $$ W_i - \I = (W_i - U_i) + (U_i - \I) = ( \ed^{\iu F_i K_i F_i} - \I ) U_i +   (U_i - \I)$$
    is uniformly rapidly local and localized near $\{\Sigma_i\}_{i \in \caI}$ by Proposition \ref{prop:rapid locality is preserved by smooth functions} and Lemma \ref{lem:*-algebras and ideals of rapidly local operators}.

    For the last claims, note that $Q_i = \chi_{\{ \abs{z-1} \leq 1/2\}}(U_i)$, so $[U_i, Q_i] = 0$ is immediate. Since $s\chi_{\{ \abs{z-1} \leq 1/2\}} = 0$ we have $F_i Q_i = Q_i F_i = 0$ and therefore $\ed^{\iu F_i K_i F_i} Q_i = Q_i \ed^{\iu F_i K_i F_i} = Q_i$. It follows that also $[W_i, Q_i] = 0$, and that $U_i Q_i = W_i Q_i$ as claimed. Finally, if $U_i|_{\caK} = \I_{\caK}$ then $F_i|_{\caK} = 0$ and therefore $W_i|_{\caK} = \I_{\caK}$.
\end{proof}

\subsection{Reducing half-space projectors}

Recall $E_a = P \Pi_a P$ and write $\Hyper_a = \Hyper_a^{(2)} = \{ x \in \R^2 \, : \, x_2 = a \}$.

\begin{proposition} \label{prop:reducing half-space projectors}
    Let $P$ be a rapidly local projector on $\H{2}{n}$ with $\Ch(P) = 0$. Then there is a uniformly rapidly local family of projectors $\{ F_a \}_{a \in \R}$ such that $F_a \leq P$ for all $a \in \R$, and the family $\{ F_a - E_a \}_{a \in \R}$ is uniformly rapidly local and localized near $\{ \Hyper_a \}_{a \in \R}$.
\end{proposition}

\begin{proof}
    Let $0 < \delta < 1/736$. It follows from Lemma \ref{lem:gapping out boundary unitaries} that there is $R < \infty$ and a family of unitaries $\{ V_a \}_{a \in \R}$ with uniform propagation bound $R$ such that $\{ V_a - \I\}_{a \in \R}$ is uniformly rapidly local and localized near $\Hyper_a$ and
    $$ \norm{U_a - V_a} \leq \delta, \quad \dist(-1, \sigma(V_a)) \geq \delta/4 $$
    for all $a \in \R$, where $\{ U_a = U_a(P) \}_{a \in \R}$ is the family of boundary unitaries associated to $P$.

    We can now apply Lemma \ref{lem:spectral gluing} to the families $\{ U_a \}$ and $\{ V_a \}$ to obtain a family of unitaries $\{ W_a \}_{a \in \R}$ such that
    $$ \norm{ U_a - W_a } \leq 2\delta, \quad \dist(-1, \sigma(W_a)) \geq \delta/8 $$
    and $\{ W_a - \I \}_{a \in \R}$ is uniformly rapidly local and localized near $\{ \Hyper_a \}_{a \in \R}$. Moreover, writing $Q_a = \chi_{ \{\abs{z-1} \leq 1/2 \} }(U_a)$ we have
    \begin{equation} \label{eq:U_a and W_a agree on a spectral window around 1}
        W_a Q_a = U_a Q_a \quad  \text{and} \quad [U_a, Q_a] = [W_a, Q_a] = 0 \quad \text{for all} \,\, a \in \R,
    \end{equation}
    and since $U_a P^{\perp} = P^{\perp}$ we also have $W_a P^{\perp} = P^{\perp}$ by the last point of Lemma \ref{lem:spectral gluing}. In particular, $[W_a, P] = 0$.

    Let $h :\T \to [0, 1]$ be a smooth function such that $h(\ed^{2 \pi \iu x}) = x$ for all $x \in (0, 1)$ for which $e^{2\pi \iu x} \in \{\abs{z-1} \geq 1/4\}$ and define
    $$ G_a := E_a - h(U_a) + h(W_a). $$

     Since $U_a = \ed^{2\pi \iu E_a}$ we have in addition to Eq. \eqref{eq:U_a and W_a agree on a spectral window around 1} that $[E_a, Q_a] = 0$. Together with $W_a Q_a = U_a Q_a$ this implies
     \begin{equation} \label{eq:G on Q}
        [G_a, Q_a] = 0 \quad \text{and} \quad G_a Q_a = E_a Q_a.
     \end{equation}
    
    We also have $\dist( 1, \sigma( U_a Q_a^{\perp} ) ) \geq 1/2$ and since $\norm{U_a - W_a} \leq 2\delta < 1/4$ we have that $\dist(1, \sigma( W_a Q_a^{\perp} )) \geq 1/4$. It follows that $E_a Q_{a}^{\perp} = h(U_a) Q_a^{\perp}$ and therefore
    \begin{equation} \label{eq:G on Q perp}
        G_a Q_a^{\perp} = h(W_a) Q_a^{\perp}.
    \end{equation}

    It follows from all this that $0 \leq G_a \leq \I$ and that $G_a$ inherits the gap from $W_a$. Indeed, by Eq. \eqref{eq:G on Q} we have $Q_a G_a Q_a = Q_a E_a Q_a$, and $Q_a$ is a spectral projector of $E_a$ for $\{ x \in [0, 1] \, : \, \abs{\ed^{2 \pi \iu x} - 1} \leq 1/2 \}$. It follows that the spectrum of $Q_a G_a Q_a$  has a distance at least $\pi^{-1} \arcsin(1/4)$ from the point 1/2. Similarly using Eq. \eqref{eq:G on Q perp} we find that $Q_a^{\perp} G_a Q_a^{\perp} = Q_a^{\perp} h(W_a) Q_a^{\perp}$ has spectrum whose distance from $1/2$ is at least $\pi^{-1} \arcsin( \delta/16 ) \geq \delta/ (16 \pi)$. We conclude that
    $$ \dist (1/2, \sigma(G_a) ) \geq \frac{\delta}{16\pi}.  $$

    Now, $\{ G_a -  E_a =  h(W_a) - h(U_a) \}_{a \in \R}$ is uniformly rapidly local and localized near $\{ \Hyper_a \}_{a \in \R}$ by Proposition \ref{prop:rapid locality is preserved by smooth functions}. Since
    $$ G_a^2 - G_a = (G_a - E_a)^2 + (G_a - E_a) E_a + E_a (G_a - E_a) + ( E_a^2 - E_a ) - (G_a - E_a) $$
    an application of Lemma \ref{lem:E is a projector away from the cut} and Lemma \ref{lem:*-algebras and ideals of rapidly local operators} yields that $\{ G_a^2 - G_a \}_{a \in \R}$ is uniformly rapidly local and localized near $\{ \Hyper_a \}_{a \in \R}$.

    Let $l : \R \to [0, 1]$ be a compactly supported smooth function which equals zero on $[0, 1/2-\delta/(16\pi)]$ and equals one on $[1/2 + \delta/(16 \pi), 1]$. Setting
    $$ F_a := l(G_a) $$
    we get that $\{ F_a \}_{a \in \R}$ is a uniformly rapidly local family of projectors. Since the smooth function $l(x) - x$ vanishes at $0$ and $1$ and is linear near those points there exists a compactly supported smooth function $m : \R \to \R$ such that
    $$ l(x) - x = (x^2 - x) m(x) $$
    for all $x \in [0, 1]$. It follows that
    $$ F_a - E_a = (G_a^2 - G_a) m(G_a) + (G_a - E_a)$$
    so $\{ F_a - E_a \}_{a \in \R}$ is uniformly rapidly local and localized near $\{ \Hyper_a \}_{a \in \R}$ as a consequence of Proposition \ref{prop:rapid locality is preserved by smooth functions} and Lemma \ref{lem:*-algebras and ideals of rapidly local operators}.

    Finally, the operators $U_a, W_a$ act as identity on $\Ran P^{\perp}$. It follows that $h(U_a) - h(W_a)$ vanishes on $\Ran P^{\perp}$ and therefore $G_a$ also acts as zero on $\Ran P^{\perp}$. Since $l(0) = 0$ the same is true for $F_a = l(G_a)$ and we conclude that $F_a \leq P$.
\end{proof}

\subsection{Exact strip decomposition}

\begin{lemma} \label{lem:estimates on the F_a}
    Under the hypotheses of Proposition \ref{prop:reducing half-space projectors}, the family of reducing half-space projectors $\{ F_a \}_{a \in \R}$ provided by that proposition satisfies the following:
    \begin{enumerate}
        \item $\{ F_a \}_{a \in \R}$ is uniformly rapidly local and localized near the upper half-planes $\{ \Upper_a \}_{a \in \R}$.
        \item $\{ P - F_a \}_{a \in \R}$ is uniformly rapidly local and localized near the lower half-planes $\{ \Lower_a \}_{a \in \R}$.
    \end{enumerate}
    Moreover, writing $C_{a, b} = (P - F_a) F_b$ and $D_{a, b} = F_a - F_b$ for any $a < b$ we have
    \begin{enumerate}
        \setcounter{enumi}{2}
        \item $ \norm{C_{a, b}} \leq D_N \, (b - a)^{-N} $ for all $N \in \N$ and constants $\{ D_N \}_{N \in N}$ which do not depend on $a$ and $b$.
        \item $\{ D_{a, b} \}_{a < b}$ is uniformly rapidly local and localized near the strips $\{ \sfS_{a, b} = \Upper_a \setminus \Upper_b \}_{a < b}$.
    \end{enumerate}
\end{lemma}

\begin{proof}
    For item 1, write
    $$ F_a = (F_a - E_a) + E_a $$
    and note that by Proposition \ref{prop:reducing half-space projectors} we have that $\{ F_a - E_a \}_{a \in \R}$ is uniformly rapidly local and localized near $\{ \Hyper_a \}_{a \in \R}$ and therefore also near $\{ \Upper_a \}_{a \in \R}$, and that $\{E_a = P \I_{\Upper_a} P \}_{a \in \R}$ is uniformly rapidly local and localized near $\{ \Upper_a \}_{a \in \R}$ by Lemma \ref{lem:*-algebras and ideals of rapidly local operators}.

    Item 2 follows similarly by writing
    $$ P - F_a = (E_a - F_a) + (P - E_a) = (E_a - F_a) + P \I_{\Upper^c_a} P $$
    and noting that $\Upper_a^c \subset \Lower_a$.
    
    For item 3, let $L = b - a$ and write $G_a = P - F_a$ and $\overline F_b = \I_{\Upper_{b - L/3}} F_b \I_{\Upper_{b - L/3}}$ and $\overline G_a = \I_{\Lower_{a + L/3}} G_a \I_{\Lower_{a + L/3}}$. It follows from items 1 \& 2 and Lemma \ref{lem:restricting close to localization region} that
    $$ \norm{ G_a - \overline G_a }, \norm{ F_b - \overline F_b } \leq D_N \, L^{-N} $$
    for some constants $\{ D_N \}_{N \in \N}$ which do not depend on $a$ or $b$. Since also $\overline G_a \overline F_b = 0$ it follows that
    $$  \norm{ C_{a, b} } = \norm{ (G_a - \overline G_a)(F_b - \overline F_b) + (G_a - \overline G_a) \overline F_b + \overline G_a (F_b - \overline F_b) + \overline G_a \overline F_b } \leq 3D_N L^{-N}.$$

    For item 4, write
    $$ D_{a, b} = (F_a - E_a) - (F_b - E_b) + (E_a - E_b) = (F_a - E_a) - (F_b - E_b) + P \I_{\sfS_{a, b}} P$$
    and note that by Proposition \ref{prop:reducing half-space projectors} we have that $\{ F_a - E_a \}_{a \in \R}$ is uniformly rapidly local and localized near $\{ \Hyper_a \}_{a \in \R}$, and therefore near $\{\sfS_{a, b} \}$, and likewise for $\{ F_b - E_b \}$. Furthermore, $P \I_{\sfS_{a, b}} P$ is uniformly rapidly local and localized near $\{ \sfS_{a, b} \}$ by Lemma \ref{lem:*-algebras and ideals of rapidly local operators}.
\end{proof}

\begin{lemma} \label{lem:strip projectors}
    Let $P$ be a rapidly local projector on $\H{2}{n}$ with $\Ch(P) = 0$. There is $L_0 < \infty$ such that for all $L \geq L_0$ there is a family of projections $\{Q_j\}_{j \in \Z}$ which is uniformly rapidly local and localized near strips $\{ \sfS_j = \Upper_{jL} \setminus \Upper_{(j+1)L} \}_{j \in \Z}$ and such that there are constants $\{ C_N \}_{N \in \N}$ with
    $$ Q_j \leq P, \quad \norm{ Q_j Q_k } \leq C_N L^{-N} \jb{j - k}^{-N} $$
    for all $N \in \N$ and all $j \neq k$ in $\Z$. Moreover, the $\{ Q_j \}$ sum strongly to a rapidly local operator $Q = \sum_j Q_j$, and
    $$ \norm{Q - P} \leq C_N L^{-N} \quad \text{for all } \,\, N \in \N. $$
\end{lemma}

\begin{proof}
    Let $\{F_a\}_{a \in \R}$ be the approximate half-plane projectors provided by Proposition \ref{prop:reducing half-space projectors} and use the notation $C_{a, b} = (P - F_a) F_b$ and $D_{a, b} = F_a - F_b$ for any $a < b$ from Lemma \ref{lem:estimates on the F_a}.

    We have
    \begin{align*}
        D_{a, b}^2 - D_{a, b} &= (F_a - F_b)^2 - (F_a - F_b) \\
        &= (\I - F_a) F_b + F_b(\I - F_a) \\
        &= C_{a, b} + C_{a, b}^*,
    \end{align*}
    where we used $F_b \leq P$ to get the last line. It follows from item 3 of Lemma \ref{lem:estimates on the F_a} that there are constants $\{ C'_N \}_{N \in \N}$ such that
    $$ \sup_{\lambda \in \sigma\big( D_{jL, (j+1)L} \big)} \dist \left( \{0, 1\}, \lambda  \right) \leq C'_N L^{-N}$$
    for all $L > 0$ and all $N \in \N$. In particular, there is $L_0 < \infty$ such that for all $L \geq L_0$ and all $j \in \Z$, the spectrum of $D_{jL, (j+1)L}$ is contained in $[-1/4, 1/4] \cup [3/4, 5/4]$. We can therefore take a compactly supported smooth function $h : \R \to [0, 1]$ which equals zero on $[-1/4, 1/4]$ and equals one on $[3/4, 5/4]$ and define projectors
    $$ Q_j := h(D_{jL, (j+1)L}). $$
    By item 4 of Lemma \ref{lem:estimates on the F_a} and Proposition \ref{prop:rapid locality is preserved by smooth functions} the family $\{Q_j\}_{j \in \Z}$ is uniformly rapidly local and localized near $\{ \sfS_j \}_{j \in \Z}$. Since $D_{jL, (j+1)L} = PD_{jL, (j+1)L}P$ and $h(0) = 0$, we also have $Q_j \leq P$ for all $j \in \Z$. Moreover,
    \begin{equation} \label{eq:norm estimate for Q_j}
        \norm{ Q_j - D_{jL, (j+1) L} } \leq C'_N L^{-N}
    \end{equation}
    for all $N \in \N$.

    Now note that for any $a < b \leq c < d$ we have $F_a F_b = F_b - C_{a, b}$ and therefore
    \begin{align*}
        D_{a, b} D_{c, d} &= F_a F_c - F_a F_d - F_b F_c + F_b F_d \\
        &= F_c - C_{a, c} - F_d + C_{a, d} - F_c + C_{b, c} + F_d - C_{b, d} \\
        &=  - C_{a, c} + C_{a, d} + C_{b, c} - C_{b, d}
    \end{align*}
    where $C_{b, b} := 0$ for the case where $b = c$. Now item 3 of Lemma \ref{lem:estimates on the F_a} implies that
    $$ \norm{D_{jL, (j+1)L} D_{kL, (k+1)L} } \leq 4 C''_N \, L^{-N} \jb{j - k}^{-N} \quad \text{for all} \,\, j \neq k. $$

    Let $f : \R \to \R$ be a compactly supported smooth function such that $h(x) = x f(x)$ for all $x \in [-1/4, 1/4] \cup [3/4, 5/4]$ (\ie $f(x) = 0$ for $x \in [-1/4, 1/4]$ and $f(x) = x^{-1}$ for $x \in [3/4, 5/4]$). Then $Q_j = D_{jL, (j+1) L} f( D_{jL, (j+1)L} ) $ and
    $$ Q_j Q_k = Q_j^* Q_k = f( D_{jL, (j+1)L} ) D_{jL, (j+1)L} D_{kL, (k+1)L} f( D_{kL, (k+1)L} )$$
    satisfies the same norm bound as $D_{jL, (j+1)L} D_{kL, (k+1)L}$ up to a constant factor.

    Finally, write $D_j = D_{jL, (j+1)L}$ and note that the sum
    $$ \sum_{j = -M}^M D_j = F_{-ML} - F_{(M+1)L} = ( F_{-ML} - E_{-ML} ) - ( F_{(M+1)L} - E_{(M+1)L} ) + P \I_{\Upper_{-ML} \setminus \Upper_{(M+1)L} } P $$
    converges strongly to $P$. Writing $Q = \sum_j Q_j$ we therefore have
    $$ Q - P = \sum_j (Q_j - D_j).$$
    By decomposing the Hilbert space as $\H{2}{n} = \bigoplus_{j \in \Z} \caH_j$ with $\caH_j = \I_{\sfS_j} \H{2}{n}$ we can apply the Schur test \ref{lem:Schur test} to $Q - P$, using the fact that $\{ Q_j \}_{j \in \Z}$ and $\{ D_j \}_{j \in \Z}$ are uniformly rapidly local and localized near $\{ \sfS_j \}_{j \in \Z}$ together with Lemma \ref{lem:distant regions} and the estimate Eq. \eqref{eq:norm estimate for Q_j} to conclude that
    $$ \norm{Q - P} \leq C_N L^{-N} $$
    for some constants $\{C_N\}_{N \in \N}$.

    Let us finally show that $Q$ is rapidly local. Recall that $\{ Q_j \}_{j \in \Z}$ is uniformly rapidly local and localized near $\{ \sfS_j \}_{j \in \Z}$ with some constants $\{ C'_N \}_{N \in \N}$. For $N \geq 2$ we have
    \begin{align*}
        \abs{Q(x, y)} &\leq \sum_j \abs{Q_j(x, y)} \leq C'_N \jb{x-y}^{-N} \, \sum_{j} \big( 1 + \d(x, \sfS_j) + \d(y, \sfS_j) \big)^{-N} \\
        &\leq C''_N \jb{x-y}^{-N}
    \end{align*}
    for constants $\{ C''_N \}_{N \in \N}$, because the sum over strips is uniformly bounded.
\end{proof}

\begin{proposition} \label{prop:exact strip decomposition}
    Let $P$ be a projector on $\H{2}{n}$ with $\Ch(P) = 0$. Assume that $P$ is rapidly local with constants $\{ C_N \}_{N \in \N}$. Then there is $L_1 < \infty$ depending only on $\{ C_N \}_{N \in \N}$ such that for all $L \geq L_1$ there is an orthogonal family of projections $\{P_j\}_{j \in \Z}$ which is uniformly rapidly local and localized near strips $\{ \sfS_j = \Upper_{jL} \setminus \Upper_{(j+1)L} \}_{j \in \Z}$, with constants $\{ D_N \}_{N \in \N}$ depending only on $\{ C_N \}_{N \in \N}$ and $n$, and such that
    $$ P = \sum_{j \in \Z} P_j. $$
\end{proposition}

\begin{proof}
    Let $\{ Q_j \}_{j \in \Z}$ be the strip projectors provided by Lemma \ref{lem:strip projectors} at length scale $L \geq L_0$, and let $\caK_j := \Ran Q_j$ and $\caK = \bigoplus_{j \in \Z} \caK_j$. Note that the spaces $\caK_j$ as subspaces of $\H{2}{n}$ need not be orthogonal to each other, so $\caK$ is in no sense a subspace of $\H{2}{n}$, it is an abstract Hilbert space.

    Since $Q_j \leq P$ we do have $\caK_j \subset \Ran P$ for all $j \in \Z$. It follows that the map
    $$ \Phi : \caK \to \H{2}{n} : (\psi_j) \mapsto \sum_j \psi_j $$
    has its image contained in $\Ran P$.
    
    Using the bounds on $\norm{Q_j Q_k}$ provided by Lemma \ref{lem:strip projectors} and the Schur test \ref{lem:Schur test} we find
    \begin{equation} \label{eq:Phi*Phi close to I}
        \norm{ \Phi^* \Phi - \I } \leq C_N L^{-N}
    \end{equation}
    for some constants $\{ C_N \}_{N \in \N}$. On the other hand, $Q := \Phi \Phi^* = \sum_j Q_j$ so
    \begin{equation} \label{eq:PhiPhi* close to P}
        \norm{ \Phi \Phi^* - P } \leq C_N L^{-N}
    \end{equation}
    as well.

    It follows from Eqs. \eqref{eq:Phi*Phi close to I} and \eqref{eq:PhiPhi* close to P} that there is $L_1 < \infty$ so that for all $L \geq L_1$, the restriction $(\Phi \Phi^*)|_{\Ran P} : \Ran P \to \Ran P$ is invertible, and therefore $\Phi$ is a bijection from $\caK$ to $\Ran P$.

    Write $T = Q - P$. Then $\I + T$ is self-adjoint and
    $$ \sup_{\lambda \in \sigma( \I + T ) }\dist \left(  1, \lambda  \right) \leq C_N L^{-N}$$
    for all $N$. We can therefore take a compactly supported smooth function $h : \R \to \R$ such that $h(x) = x^{-1/2}$ on the spectrum of $\I + T$ and put
    $$ H = (\I + T)^{-1/2} = h(\I + T). $$
    This operator equals $(\Phi \Phi^*)^{-1/2}$ on $\Ran P$, and equals $\I$ on $\Ran P^{\perp}$, and it is rapidly local by Proposition \ref{prop:rapid locality is preserved by smooth functions}. Setting $\Psi = H \Phi$ it follows that
    $$ \Psi^* \Psi = \I_{\caK}, \quad \quad \Psi \Psi^* = P.$$
    That is, $\Psi : \caK \to \Ran P$ is a unitary map, and it sends each $\caK_j$ to the subspace $H \Ran Q_j \subset \H{2}{n}$. These subspaces are therefore orthogonal to each other, so the projectors
    $$ P_j := \Psi \I_{\caK_j} \Psi^* = H Q_j H $$
    are orthogonal to each other and sum to $P$.

    Finally, since $\{ Q_j \}_{j \in \Z}$ is uniformly rapidly local and localized near $\{ \sfS_j \}_{j \in \Z}$ by Lemma \ref{lem:strip projectors}, so is the family $\{ P_j \}_{j \in \Z}$ by Lemma \ref{lem:*-algebras and ideals of rapidly local operators}.
\end{proof}

\subsection{Rapidly local Wannier bases in 1D} \label{subsec:1D rapidly local Wannier}

\begin{theorem} \label{thm:1D rapidly local Wannier}
    Let $P$ be a projector on $\H{2}{n}$ which is rapidly local and localized near the line $\Hyper_a$ with constants $\{ C_N \}_{N \in \N}$. Then $P$ admits a rapidly local Wannier basis with constants $\{ D_N \}_{N \in \N}$ depending only on the $\{ C_N \}_{N \in \N}$ and on $n$.
\end{theorem}

\begin{proof}
    We will want to apply Proposition \ref{prop:exact strip decomposition} to cut $P$ transversally to its localization axis. To this end, first rotate $P$ by $90^{\circ}$ and translate so that $P$ is rapidly local and localized near the vertical line $\sfV = \{ x \in \R^2 \, : \, x_1 = 0 \}$.

    Since $P$ is localized near $\sfV$, we have $\Ch(P) = 0$. To see this, recall the boundary unitaries $U_a = \ed^{2\pi \iu E_a}$ with $E_a = P \Pi_a P$, where $\Pi_a = \I_{\Upper_a}$ is the projection onto the upper half plane. We have $E_a^2 - E_a = P \Pi_a [P, \Pi_a] P$, and since $P$ is localized near $\sfV$, it follows that $[P, \Pi_a]$ is localized near the intersection $x_a = \{(0, a)\} = \sfV \cap \Hyper_a$. Indeed, $([ P, \Pi_a ])(x, y) = 0$ unless either $x \in \Upper_a$ and $y \notin \Upper_a$ or vice versa. We get in either case
    \begin{align*}
        \abs{ ([P, \Pi_a])(x, y)} &= \abs{ P(x, y) } \leq C_{2N} \, \jb{x-y}^{-2N} \jb{ \d(x, \sfV) + \d(y, \sfV) }^{-2N} \\
        &\leq C_{2N} \jb{x-y}^{-N} \jb{ \d(x, \Hyper_a) + \d(y, \Hyper_a) }^{-N} \jb{ \d(x, \sfV) + \d(y, \sfV) }^{-2N} \\
        &\leq C_{2N} \jb{x-y}^{-N} \jb{ \abs{x - x_a} + \abs{y - x_a} }^{-N}
    \end{align*}
    where we used Eq. \eqref{eq:triangle for the Japanese bracket} and $\abs{ x - x_a} + \abs{y - x_a} \leq \d(x, \sfV) + \d(y, \sfV) + \d(x, \Hyper_a) + \d(y, \Hyper_a)$ to obtain the third line. Since this holds for all $N \in 
    \N$, we find that $\{[P, \Pi_a]\}_{a \in \R}$ is uniformly rapidly local and localized near $\{ x_a \}_{a \in \R}$, and the same is true for $\{ E_a^2 - E_a \}_{a \in \R}$ by Lemma \ref{lem:*-algebras and ideals of rapidly local operators}. By repeating the proof of Proposition \ref{prop:boundary unitaries act near the boundary} we conclude that also $\{ U_a - \I \}_{a \in \R}$ is uniformly rapidly local and localized near the points $\{ x_a \}_{a \in \R}$.

    Recall from Section \ref{subsec:flow and chern number} that the flow of $U_a$ is given by
    $$ \Flow(U_a) = \dim \ker (A_{a, b} - \I) - \dim \ker (A_{a, b} + \I) $$
    where $A_{a, b} = U_a \I_{\Right_b} U_a^* - \I_{\Right_b} = [U_a, \I_{\Right_b}] U_a^*$ for any $b \in \R$. Since $U_a-\I$ is localized near $x_a = \{(0, a)\}$, it follows from Lemma \ref{lem:restricting close to localization region} that
    $$ \norm{ U_a - U_a^{(r)} } \leq D_N r^{-N}, \quad \text{where} \quad U_a^{(r)} = \I_{\{ \abs{x - x_a} \leq r \}} U_a \I_{\{ \abs{x - x_a} \leq r \}} + \I_{\{ \abs{x - x_a} \leq r \}}^{\perp} $$
    for some constants $\{ D_N \}_{N \in \N}$ and all $r > 0$. For $r$ large enough, and $b$ large enough in function of $r$, we then have $[ U_a^{(r)}, \I_{\Right_b} ] = 0$ and therefore
    $$ \norm{A_{a, b}} = \norm{ [ (U_a - U_a^{(r)}), \I_{\Right_b} ] U_a^* } \leq D_N r^{-N} < 1$$
    and so the kernels of $A_{a, b} \pm \I$ are trivial. In particular, $\Flow(U_a) = 0$. By the bulk-boundary correspondence \ref{prop:bulk-boundary correspondence} also $\Ch(P) = 0$.

    We can now apply Proposition \ref{prop:exact strip decomposition} to get an $L < \infty$ and a decomposition
    $$ P = \sum_j P_j $$
    with $\{ P_j \}_{j \in \Z}$ an orthogonal family of projectors which is uniformly rapidly local and localized near $\{ \sfS_j = \Upper_{jL} \setminus \Upper_{(j+1) L} \}$. In particular, $\{ P_j \}_{j \in \Z}$ is uniformly rapidly local and localized near the lines $\{ \Hyper_{jL} \}$. Since $P_j = P P_j$ for all $j \in \Z$ and $P$ is rapidly local and localized near $\sfV$ it follows from Lemma \ref{lem:*-algebras and ideals of rapidly local operators} that $\{ P_{j} \}_{j \in \Z}$ is uniformly rapidly local and localized near $\sfV$ as well. By the characterization of Lemma \ref{lem:norm characterization of rapid locality} this means that all the norms $\norm{P_j}_{N, \Hyper_{jL}}, \norm{P_j}_{N, \sfV} < C_N$ are uniformly bounded. It follows that
    \begin{align*}
        \sum_y \abs{ P_j(x, y) } &\jb{x-y}^{N} \jb{ \d(x, x_{jL}) + \d(y, x_{jL}) }^{N} \\
        &\leq \sum_y  \abs{ P_j(x, y) } \jb{x-y}^{N} \jb{ \d(x, \sfV) + \d(x, \Hyper_{jL}) + \d(y, \sfV) + \d(y, \Hyper_{jL}) }^{N} \\
        &\leq \sum_y \abs{P_j(x, y)} \jb{x-y}^N \, 2^{N-1}\left(  \jb{ \d(x, \sfV) + \d(y, \sfV) }^N + \jb{ \d(x, \Hyper_{jL}) + \d(y, \Hyper_{jL})}^N  \right) \\
        & \leq 2^{N-1} \left( \norm{P_j}_{N, \sfV} + \norm{P_j}_{N, \Hyper_{jL}} \right) < 2^N C_N
    \end{align*}
    and likewise for the sum over $x$. It follows that the norms $\norm{P_j}_{N, \{x_{jL}\}} \leq 2^N C_N$ are uniformly bounded and therefore $\{ P_j \}_{j \in \Z}$ is uniformly rapidly local and localized near $\{ x_{jL} \}_{j \in \Z}$ by Lemma \ref{lem:norm characterization of rapid locality}. This means that there are constants $\{ D'_N \}_{N \in \N}$ which depend only on the $\{ C_N \}$ and on $n$ so that
    $$ \abs{P_j(x, y)} \leq D'_N \, \jb{x-y}^{-N} \jb{ \abs{x - x_{jL}} + \abs{y - x_{jL}} }^{-N}$$
    for all $N \in \N$. For any unit vector $\psi \in \Ran P_j$ we then have
    \begin{align*}
        \abs{\psi(x)} &= \abs{ \sum_y P_j(x, y) \psi(y) } \leq \sum_y \abs{P_j(x, y)} \\
        &\leq D'_{N + 2} \jb{x - x_{jL}}^{-N-2} \sum_y \jb{x-y}^{-N-2} \leq  D_N \jb{x - x_{jL}}^{-N}
    \end{align*}
    for constants $\{ D_N \}_{N \in \N}$ depending only on the $\{ D'_N \}$.

    In particular, if $\{ w^{(j)}_{\al} \}_{\al \in \caA_j}$ are orthonormal bases for each $\Ran P_j$, then these bounds hold for each basis vector $w^{(j)}_{\al}$, and therefore the union of all these bases forms a rapidly localized Wannier basis for $P$ with constants $\{ D_N \}_{N \in \N}$ depending only on the locality constants $\{C_N\}_{N \in \N}$ of $P$ and on $n$.
\end{proof}

\subsection{Rapidly local Wannier bases in 2D} \label{subsec:2D rapidly local Wannier}

\begin{theorem} \label{thm:2D rapidly local Wannier}
    Let $P$ be a rapidly local projector on $\H{2}{n}$ with $\Ch(P) = 0$. Then $P$ admits a rapidly local Wannier basis.
\end{theorem}

\begin{proof}
    By Proposition \ref{prop:exact strip decomposition} we can decompose
    $$ P = \sum_{j \in \Z} P_j $$
    where $\{ P_j \}_{j \in \Z}$ is a uniformly rapidly local family of mutually orthogonal projectors that is localized near strips $\{ \sfS_j = \Upper_{jL} \setminus \Upper_{(j+1)L} \}$ for some $L < \infty$.

    In particular, $\{ P_j \}_{j \in \Z}$ is uniformly rapidly local and localized near the lines $\{ \Hyper_{jL} \}_{j \in \Z}$. It therefore follows from Theorem \ref{thm:1D rapidly local Wannier} that there are constants $\{D_N\}_{N \in \N}$ such that each $P_j$ admits a rapidly local Wannier basis $\{ w^{(j)}_{\al} \}_{\al \in \caA_j}$ with constants $\{ D_N \}_{N \in \N}$. But then the union of all these Wannier bases is a rapidly local Wannier basis for $P$ with constants $\{ D_N \}_{N \in \N}$.
\end{proof}

\section{From rapidly local to exponentially local Wannier bases} \label{sec:rapid to exponential}

Let us first show that the localization centers of a rapidly localized Wannier basis cannot cluster too much. See also \cite[Lemma 2.1.]{lu2024algebraic} for an analogous result in the continuum.

\begin{lemma} \label{lem:no clustering}
    Let $P$ be a rapidly local projector on $\H{d}{n}$ and let $\{ w_{\al} \}_{\al \in \caA}$ be a rapidly local Wannier basis for $P$ with localization centers $\{ x_{\al} \}_{\al \in \caA}$. Then there is a constant $C < \infty$ such that
    $$ \# \{ \al \in \caA \, : \, \abs{x_{\al} - c} \leq R \} \leq C (1 + R)^d$$
    for any $c \in \R^d$ and any $R \geq 0$.
\end{lemma}

\begin{proof}
    Let $\{ w_{\al} \}_{\al \in \caA}$ be rapidly localized with constants $\{ C_N \}_{N \in \N}$.

    For any $c \in \R^d$ and $R \geq 0$ write
    $$  \caA_{c, R} = \{ \al \in \caA \, : \, \abs{x_{\al} - c} \leq R \} \quad \text{and} \quad B_{c, R} = \{ x \in \Z^d \, : \, \abs{x - c} \leq R \}. $$

    We have for any $R > 0$, any $\al \in \caA_{c, R}$, and any $N \geq (d+1)/2$ that
    \begin{align*}
        \sum_{x \in \Z^d \setminus B_{c, 2R}} \abs{ w_{\al}(x) }^2 &\leq \sum_{x \in \Z^d \setminus B_{x_{\al, R}}} \abs{w_{\al}(x)}^2 \leq C_N^2 \sum_{x \in \Z^d \setminus B_{x_{\al}, R}} \, \jb{ x - x_{\al} }^{-2N} \\
        &\leq C_{N, d} R^{d - 2N}
    \end{align*}
    where $C_{N, d} < \infty$ depends only on $C_{N}$ and on $d$.

    Take $N = \ceil{(d+1)/2}$ and $R \geq 1$, then since $\sum_{x \in \Z^d} \abs{w_{\al}(x)}^2 = 1$ there is a constant $C_d < \infty$ such that
    $$ \sum_{x \in B_{c, 2R} } \abs{w_{\al}(x)}^2 > 1- C_d R^{-1} $$
    for all $\al \in \caA_{c, R}$, and therefore
    $$ \sum_{\al \in \caA_{c, R}}  \sum_{x \in B_{c, 2R} } \abs{w_{\al}(x)}^2 > \# \caA_{c, R} (1- C_d R^{-1}).  $$

    On the other hand, by the Parseval identity we have
    $$ \sum_{\al \in \caA_{c, R}} \abs{w_{\al}(x)}^2 \leq n$$
    for any $x \in \Z^d$ and therefore
    $$ \sum_{\al \in \caA_{c, R}}  \sum_{x \in B_{c, 2R} } \abs{w_{\al}(x)}^2 \leq n D_{d} R^d  $$
    for some constant $D_d < \infty$.

    Comparing the two bounds, we get
    $$ \# \caA_{c, R} < 2 n D_d R^d $$
    whenever $R \geq R_0 := \max \{ 1, 2C_d \}$. The claim of the lemma now holds with $C = 2 n D_d R_0^d$.
\end{proof}

The following theorem establishes that if an exponentially local projection admits a rapidly localized Wannier basis, then it actually admits an exponentially localized Wannier basis. See \cite{lu2021algebraic} for a similar result in the continuum.

\begin{theorem} \label{thm:rapid to exponential}
    Let $P$ be an exponentially local projector on $\H{d}{n}$ which admits a rapidly local Wannier basis. Then $P$ admits an exponentially local Wannier basis.
\end{theorem}

\begin{proof}
    Let $\{ w_{\al} \}_{\al \in \caA}$ be a rapidly local Wannier basis for $P$ with constants $\{ C_N \}_{N \in \N}$ and localization centers $\{ x_{\al} \}_{\al \in \caA}$.

    Let $\Phi : l^2(\caA) \to \H{d}{n}$ be the isometry defined by $\Phi e_{\al} = w_{\al}$. Then
    $$ \Phi^* \Phi = \I, \quad \text{and} \quad  \Phi \Phi^* = P.$$

    Fix $R \geq 1$ and define truncations
    $$ \tilde w_{\al}(x) = \begin{cases}
        w_{\al}(x) \quad &\text{if} \,\, \abs{x - x_{\al}} \leq R \\
        0 \quad &\text{otherwise}
    \end{cases} $$
    and $\widetilde \Phi : l^2(\caA) \to \H{d}{n}$ defined by $\widetilde \Phi e_{\al} = \tilde w_{\al}$ for all $\al \in \caA$.

    Decompose $l^2(\caA)$ as a direct sum over its canonical basis, and $\H{d}{n} = \bigoplus_{ x \in \Z^d } \caH_x$ with each $\caH_x \simeq \C^n$ the on-site degree of freedom at $x \in \Z^d$. With respect to these decompositions, the operator $\Delta \Phi := \Phi - \widetilde \Phi$ has blocks $\Delta\Phi(x, \al)$ given by
    $$ \Delta \Phi(x, \al) e_{\al} = w_{\al}(x) - \tilde w_{\al}(x). $$
    The column sums can be uniformly bounded by
    \begin{align*}
        \sum_{x \in \Z^d} \norm{ \Delta \Phi(x, \al) } &= \sum_{x \in \Z^d} \abs{ w_{\al}(x) - \tilde w_{\al(x)} } = \sum_{x \, : \, \abs{x - x_{\al}} > R } \abs{w_{\al}(x)} \\
        &\leq C_{N}  \sum_{x \, : \, \abs{x - x_{\al}} > R } \jb{x - x_{\al}}^{-N} \leq D'_N (1+R)^{d-N}
    \end{align*}
    for any $N \geq d+1$ and where $\{ D'_N \}_{N \geq d+1}$ are constants depending only on $\{ C_N \}_{N \geq d+1}$ and $d$. 
    
    Writing $\caA_{c, r} = \{ \al \in \caA \, : \, \abs{x_{\al} - c} \leq r \}$ Lemma \ref{lem:no clustering} provides a constant $C$ such that 
    $$\# \caA_{c, r} \leq C (1 + r)^d$$
    for any $c \in \R^d$ and any $r \geq 0$. Fix $x \in \Z^d$ and write $\caB_k = \caA_{x, (k+1)R} \setminus \caA_{x, kR}$ for any $k \in \N$. The row sums of $\Delta \Phi$ can then be bounded for any $N \geq d+1$ as
    \begin{align*}
        \sum_{\al \in \caA} \norm{ \Delta \Phi(x, \al) } &= \sum_{\al \in \caA} \abs{ w_{\al}(x) - \tilde w_{\al}(x) } = \sum_{k = 1}^{\infty} \sum_{\al \in \caB_k} \abs{ w_{\al}(x) } \\
        &\leq C' \sum_{k = 1}^{\infty} (1 + 2R)^d k^{d-1} \,\, C_N \, (1 + kR)^{-N} \\
        &\leq D''_N (1 + R)^{d-N}
    \end{align*}
    for some constants $\{ D''_N \}_{N \geq d +1}$ depending only on $\{ C_N \}_{N \geq d +1}$ and on $d$ and $n$. Here, $C'$ is a constant so that each $\caB_k$ can be covered by at most $C' k^{d-1}$ sets of the form $\caA_{c, R}$. It follows from these bounds on the row and column sums and the Schur test \ref{lem:Schur test} that
    $$ \norm{ \Delta \Phi } \leq D_N (1 + R)^{d - N} $$
    for all $N \geq d + 1$, where $D_N = \max \{ D'_N, D''_N \}$.

    Take $R$ large enough so that $\norm{ \Delta \Phi } < 1/4$ and set
    $$ \Psi = P \widetilde \Phi \quad \text{and} \quad H = \Psi^* \Psi = \widetilde \Phi^* P \widetilde \Phi. $$
    The map $\Psi : l^2(\caA) \to \Ran P$ is bijective because $\Psi = \Phi \Phi^* \widetilde \Phi$ and $\Phi$ is unitary from $l^2(\caA)$ to $\Ran P$ while $\Phi^* \widetilde \Phi$ is invertible since $\norm{\Phi^* \widetilde \Phi - \I} = \norm{\Phi^* \, \Delta \Phi} < 1/4$.
    
    Moreover, since $P \Phi = \Phi$ we have $\Phi - \Psi = P \Delta \Phi$ and therefore $\norm{ \Psi - \Phi } < 1/4$. Since $\Phi$ is an isometry onto $\Ran P$, it follows that $3/4 \norm{\xi} \leq \norm{\Psi \xi} \leq 5/4 \norm{\xi}$ for every $\xi \in l^2(\caA)$ and so $(3/4)^2 \I \leq H \leq (5/4)^2 \I$. We can therefore define 
    $$ \widetilde \Psi = \Psi H^{-1/2},
    $$ 
    which satisfies
    \begin{equation} \label{eq:widetilde Psi is spot on}
        \widetilde \Psi^* \widetilde \Psi = \I \quad \text{and} \quad \widetilde \Psi \widetilde \Psi^* = P.
    \end{equation}

    It follows from exponential locality of $P$ that there are constants $C, C'$ and $\mu$ depending on our choice of $R$ such that
    $$ \norm{\Psi(x, \al)} \leq C \ed^{-\mu \abs{x - x_{\al}}} \quad \text{and} \quad \norm{H(\al, \beta)} \leq C' \ed^{-\mu \abs{x_{\al} - x_{\beta}}}. $$
    Together with the no-clustering bound of Lemma \ref{lem:no clustering} this implies that if we think of $\caA$ as a graph equipped with metric $\d(\al, \beta) = \max\{\abs{x_{\al} - x_{\beta}}, 1-\delta_{\al, \beta} \}$ then $H$ satisfies the assumption of the Combes-Thomas estimate of \cite[Theorem 10.5]{aizenman2015random}. We therefore have
    $$ \abs{ (H - z)^{-1}(\al, \beta) } \leq \frac{2}{\delta} \ed^{- \mu' \min\{ \delta, 1\} \abs{ x_{\al} - x_{\beta} }} $$
    where $\delta = \dist( z, \sigma(H) )$. Using the holomorphic functional calculus we can write
    $$ H^{-1/2} = \frac{\iu}{2 \pi} \oint_{\Gamma} \dd z \, \frac{ z^{-1/2} }{H - z} $$
    where the contour $\Gamma$ lies in the domain of analyticity of $z^{-1/2}$, encircles the spectrum of $H$, and stays a bounded distance away from its spectrum. It follows that there are constants $C'', \mu''$ such that
    $$ \norm{ H^{-1/2}(\al, \beta) } \leq C'' \ed^{-\mu'' \abs{x_{\al} - x_{\beta}}}. $$
    Finally, using a triangle inequality and no-clustering Lemma \ref{lem:no clustering}, we conclude that there are constants $C''', \mu'''$ such that
    $$ \norm{ \widetilde \Psi(x, \al) } \leq C''' \ed^{-\mu''' \abs{x - x_{\al}}}. $$

    Now define
    $$ v_{\al} :=  \widetilde \Psi \, e_{\al} \quad \text{for all } \,\, \al \in \caA.$$
    Then $\{ v_{\al} \}_{\al \in \caA}$ is an orthonormal basis for $\Ran P$ by Eq. \eqref{eq:widetilde Psi is spot on}, and
    $$ \abs{v_{\al}(x)} = \norm{ \widetilde \Psi(x, \al) } \leq C''' \ed^{-\mu''' \abs{x - x_{\al}}}, $$
    so $\{ v_{\al} \}_{\al \in \caA}$ is an exponentially localized Wannier basis for $P$.
\end{proof}

\section{A rapidly localized Wannier basis implies Chern triviality} \label{sec:rapid Wannier implies Ch(P) = 0}

The following theorem proves that if a rapidly local projector admits a rapidly localized Wannier basis, then its Chern number vanishes. Results of the same type have been obtained in the continuum setting in \cite{marcelli2023localization, lu2024algebraic}, see also \cite{rossi2025algebraic}. The first two of these works prove the result by working with the Kubo formula for the Chern number. Here we give an alternative strategy using the formulation in terms of an index of a pair of projections through the flux insertion operator $F_a$. The localized Wannier basis then allows the construction of a modified flux insertion $\widetilde F_a$ which leaves the projector invariant, and such that $PF_aP - \widetilde F_a$ is compact on $\Ran P$. A similar construction has appeared in \cite[~Proposition 4.3]{bols2023fredholm}.

\begin{theorem} \label{thm:rapidly local Wannier basis implies Ch(P) = 0}
    Let $P$ be a rapidly local projector on $\H{2}{n}$. Suppose $P$ admits a rapidly localized Wannier basis. Then $\Ch(P) = 0$.
\end{theorem}

\begin{proof}
    The Chern number of $P$ can be expressed as a Fredholm index
    $$ \Ch(P) = \ind(P, F_a) = \ind_F( P F_a P : \Ran P \to \Ran P ), $$
    see for example \cite[~Theorem 5.2.]{avron1994index}.

    Let $\{ w_{\al} \}_{\al \in \caA}$ be a rapidly local Wannier basis for $P$ with constants $\{ C_N \}_{N \in \N}$ and localization centers $\{ x_{\al} \}_{\al \in \caA}$.

    For the remainder of this proof we identify $\R^2$ with $\C$ under the map $x \to x_1 + \iu x_2$.
    
    Since $\{ x_{\al} \}_{\al \in \caA}$ is a countable subset of $\R^2$ (by Lemma \ref{lem:no clustering}), we can take $a \in \C \setminus \Z^2$ so that none of the $x_{\al}$ is equal to $a$. Define a modified flux insertion operator $\widetilde F_a$ by
    $$
    \widetilde F_a w_{\al} = \phi( x_{\al}) w_{\al} \,\, \text{for all} \,\, \al \in \caA \quad \text{and} \quad \widetilde F_a P^{\perp} = 0,
    $$
    where
    $$
    \phi(x) = \frac{x - a }{\abs{x-a}}
    $$
    for any $x \neq a$. Then $\widetilde F_a^* \widetilde F_a = \widetilde F_a \widetilde F_a^* = P$. With this notation the original flux insertion operator is given by
    $$ F_a \delta_{x, i} = \phi(x) \delta_{x, i} $$
    for all $x \in \Z^2$ and $i = 1, \cdots, n$.

    Since $\widetilde F_a : \Ran P \to \Ran P$ is unitary, it is a Fredholm operator with vanishing Fredholm index. Since the Fredholm index is stable under compact perturbations, the theorem will be proven once we show that $P F_a P - \widetilde F_a$ is compact on $\Ran P$.

    The no clustering Lemma \ref{lem:no clustering} implies that the sums
    $$ M_1 := \sup_{x \in \Z^2} \sum_{\al \in \caA} \jb{ x - x_{\al} }^{-3} < \infty, \quad \text{and} \quad  M_2 := \sup_{c \in \R^2} \sum_{x \in \Z^2} \jb{c - x}^{-3} < \infty$$
    are uniformly bounded.

    Note the inequality
    \begin{equation} \label{eq:compactness inequality}
        \abs{ \phi(x) - \phi(c) } \leq \frac{2 \abs{ x - c }}{\abs{c - a}}
    \end{equation}
    for any $c, x \neq a$.

    Let $\Phi : l^2(\caA) \to \H{2}{n}$ be the isometry given by $\Phi e_{\al} = w_{\al}$ and let $D e_{\al} = \phi(x_{\al}) e_{\al}$ be the flux operator on $l^2(\caA)$. Then $\Phi^* \Phi = \I$, $\Phi \Phi^* = P$, and we can write the modified flux insertion operator as
    $$ \widetilde F_a = \Phi D \Phi^*. $$
    The operator $K = F_a \Phi - \Phi D : l^2(\caA) \to \H{2}{n}$ has matrix blocks
    $$ K(x, \al) = (\phi(x) - \phi(x_{\al}) ) w_{\al}(x).  $$
    Let $E_R$ be the projection onto $\Span \{ e_{\al} \,: \, \abs{x_{\al} - a} \leq R \}$. For $\al$ such that $\abs{x_{\al} - a} > R$ we can use Eq. \eqref{eq:compactness inequality} to bound the column sums by
    \begin{align*}
        \sum_{x \in \Z^2} \norm{ KE_R^{\perp}(x, \al) } &= \sum_{x \in \Z^2} \abs{ \phi(x) - \phi(x_{\al}) } \abs{ w_{\al}(x) } \\
        &\leq C_4 \sum_{x \in \Z^2} \frac{2 \abs{x - x_{\al}}}{\abs{x_{\al} - a}} \jb{x - x_{\al}}^{-4} \\
        &\leq \frac{2C_4 M_2}{R}.
    \end{align*}
    Similarly, the row sums are bounded as
    \begin{align*}
        \sum_{\al \in \caA} \norm{ K E_R^{\perp}(x, \al) } &= \sum_{\al : E_R e_{\al} = 0} \abs{ \phi(x) - \phi(x_{\al}) } \abs{ w_{\al}(x) } \\
        &\leq C_4 \sum_{\al : E_R e_{\al} = 0} \, \frac{2 \abs{x - x_{\al}}}{\abs{x_{\al} - a}} \jb{x - x_{\al}}^{-4} \\
        &\leq \frac{2 C_4 M_1}{R}.
    \end{align*}
    By the Schur test \ref{lem:Schur test} we conclude that
    $$\norm{ K E_R^{\perp} } \leq \frac{2 C_4 M}{R}$$
    with $M = \max\{ M_1, M_2 \}$. In other words, $K$ can be approximated arbitrarily well in norm by the compact operators $K E_R$, so $K$ is compact.

    We finally note that
    $$ PF_aP - \widetilde F_a = P K \Phi^* $$
    and conclude that $PF_aP - \widetilde F_a$ is compact as an operator on $\Ran P$, as required.
\end{proof}

\begin{remark}
    In all even dimensions $d$ there is a well-defined higher Chern number $\Ch_d(P) = \ind_F( (P \otimes \I) U^{(d)}_a (P \otimes \I) : \Ran (P \otimes \I) \to \Ran (P \otimes \I) )$ with $U_a^{(d)}$ the Dirac unitary which generalizes the flux insertion operators to arbitrary even dimensions. See for example \cite{prodan2016bulk, katsura2018noncommutative}. The proof strategy of Theorem \ref{thm:rapidly local Wannier basis implies Ch(P) = 0} can be easily adapted to prove that if a rapidly local projector $P$ on $\H{d}{n}$ admits a rapidly localized Wannier basis, then $\Ch_d(P) = 0$. Note that the no clustering Lemma \ref{lem:no clustering} holds in arbitrary dimension.
\end{remark}


\bibliographystyle{unsrturl}
\bibliography{bib}

\end{document}